\documentclass[12pt]{article}
\usepackage[margin=1in]{geometry}
\usepackage{graphicx} \usepackage{multirow}
\usepackage{amsmath,amsthm,amssymb}
\usepackage[pagewise]{lineno}
\usepackage{natbib}
\usepackage{setspace}
\usepackage{url}
\usepackage[colorlinks=true,linkcolor=blue,citecolor=blue]{hyperref}
\usepackage[table]{xcolor}
\usepackage{thmtools}

\newcommand{\p}[1]{\left(#1\right)}
\renewcommand{\b}[1]{\left[#1\right]}

\newcommand{\EE}[2]{\mathbb{E}_{#1}\b{#2}}
\def\mc{\mathcal}

\def\given{\; | \;}
\def\mujall{\hat \mu_j^{\text{all}}}
\def\muoall{\hat \mu_1^{\text{all}}}
\def\mutall{\hat \mu_2^{\text{all}}}

\def\PH{\mathbb{P}_{\mc H_0}}
\def\PX{\mathbb{P}_{\mc X}}
\def\PXk{\mathbb{P}_{\vX}}
\def\Pboth{\mathbb{P}_{\mc H_0, \mc X}}
\def\Pbothk{\mathbb{P}_{\mc H_0, \vX}}
\def\mono{\textnormal{mono}}
\def\poly{\textnormal{poly}}

\def\allwrong{Z_{\textnormal{all-wrong}}}
\def\somewrong{Z_{\textnormal{some-wrong}}}
\def\allright{Z_{\textnormal{all-right}}}
\def\asto{\overset{a.s.}{\longrightarrow}}

\def\vX{\vec{\mc X}}
\def\hist{\mc H_0}

\usepackage{array}

\usepackage{float}

\definecolor{rowgray}{gray}{0.9} 

\usepackage{makecell}

\newtheorem{theorem}{Theorem}

\newtheorem{lemma}[theorem]{Lemma}

\usepackage{cleveref}
\crefname{objection}{Objection}{Objections}
\Crefname{objection}{Objection}{Objections}

\title{Algorithmic Monoculture and its Critics\thanks{Authors contributed
equally. For helpful feedback and discussion, we would like to thank Kathleen
Creel, Garrett Cullity, Kevin Dorst, Rachel Fraser, Nikhil Garg, Daniel Greco,
Toby Handfield, Alan H\'{a}jek, Caspar Hare, Ren\'ee J\o rgensen, Harvey
Lederman, Daniel Mu\~{n}oz, Jacob Nebel, Kenny Peng, Alex Slivkins, Jack Spencer, Daniel Wodak, Patrick Wu, and audiences at MIT, Yale, and the Ranch Metaphysics Workshop.}}
\author{Brian Hedden \\ MIT \and Manish Raghavan \\ MIT}
\date{}

\begin{document}

\maketitle

\begin{abstract}
Algorithmic decision-making is replacing idiosyncratic human judgment in domains
such as hiring, lending, and criminal justice. This shift promises increased
consistency, but many scholars worry that it can go too far. They warn of the
dangers of algorithmic monoculture, in which all decisions across a domain are
made using a single algorithm. We systematically evaluate a range of objections
to monoculture, formalizing and rigorously assessing familiar critiques
alongside novel ones. These objections concern systemic exclusion, agency and
gaming, and information aggregation and exploration.
We conclude that
monoculture is less problematic than its critics have supposed: commonly cited
objections fail, and while other objections have some force, they are not decisive against monoculture in general.
\end{abstract}

\section{Introduction}

Algorithmic decision-making tools are replacing idiosyncratic human judgment in domains such as hiring, lending, medicine, and criminal
justice.
 This shift has the potential to reduce bias and to increase
consistency~\citep{kleinberg2020algorithms}. 
Reduced bias would be great, but many people doubt whether it is
likely.\footnote{For seminal discussion, see \citet{barocas2016big}.} Increased
consistency, on the other hand, is more likely, and certainly achievable. But
can we get too much of a good thing? Many scholars, including philosophers,
legal scholars, and computer scientists, think so. They warn about the dangers
of \textit{algorithmic monoculture}, in which all decisions of a certain type
are made using the same algorithm \citep{kleinberg2021algorithmic,
creel2022algorithmic, bommasani2022picking, toups2023ecosystem,
jain2024algorithmic, CreelManuscript-CREAMA}.\footnote{\citet{peng2024monoculture}  raise worries about monoculture in hiring which we discuss in \S\ref{subsec:aggregation} below, but they also point out some benefits of monoculture.} 

In this vein, \citet{creel2022algorithmic} raise the alarm about the rise of `algorithmic leviathans.'\footnote{They go so far as to propose `making it illegal for one company or one company’s algorithm to dominate an entire sector of hiring or lending' (pp. 35-6).} \citet{oneil2016weapons} describes `weapons of math destruction,' which allegedly are especially harmful when deployed at scale. And \citet[1]{bommasani2022picking} worry that monoculture will yield `outcome homogenization' and `institutionalize systemic exclusion and reinscribe social hierarchy.'

Consider the obvious analogy with agriculture. There, monoculture, in which we plant only a single crop species, increases the risk of total crop failure. Better to have polyculture, so that even if one crop fails due to pests or disease, it's unlikely that all of them will.\footnote{Similar analogies hold in the context of computer security
\citep{birman2009monoculture}.}

But we shouldn't be convinced about the badness of algorithmic monoculture just by a quick analogy with agriculture. Nor, of course, should we just trust our gut reactions. 

We evaluate algorithmic monoculture by considering three objections to it. The
first and most prominent (\S\ref{sec:exclusion}) says that monoculture threatens
to systemically exclude certain people or groups from valuable opportunities.
The second (\S\ref{sec:agency}) says that monoculture gives people too little
agency in some respects (making it harder for them to work for a better outcome)
and too much agency in others (making it too easy for them to `game the system'
and improve their prospects without genuinely improving their merit). We argue
that these two objections are not compelling. We doubt whether monoculture
raises the risk of systemic exclusion, and considerations of agency point in
both directions, some favoring algorithmic polyculture (in which different decisions are made using different algorithms) and others favoring monoculture.

A more compelling objection (\S\ref{sec:information}) says that monoculture is inferior to polyculture from an informational perspective. This is because (\S\ref{subsec:aggregation}) it fails to take advantage of the `wisdom of crowds' that we could get by deploying a diverse set of algorithms, and because (\S\ref{subsec:exploration}) it does poorly on the explore-exploit trade-off, tending to favor `known quantities' rather than exploring new types of candidates about whom there is more uncertainty. We show that such informational considerations really do favor polyculture over monoculture, provided that the sole algorithm deployed under monoculture is on a par with each of the various algorithms deployed under polyculture. However, a more sophisticated monoculture built out of polyculture by packaging together all of the latter's various algorithms into a single `ensemble' algorithm can perform at least as well as the original polyculture, and often better. But it is unclear whether such ensembling is practically feasible.

Throughout our discussion, we focus primarily on the domain of hiring, though we also briefly consider other domains such as lending, generative AI, and drug discovery in \S\S\,\ref{sec:information}--\ref{sec:discussion}.

\section{Exclusion}
\label{sec:exclusion}
The most prominent objection to monoculture, advanced by Creel and a series of co-authors \citep{creel2022algorithmic, bommasani2022picking, toups2023ecosystem, CreelManuscript-CREAMA}, is that monoculture threatens to systemically exclude certain people from valuable opportunities like employment. At its core, the worry is intuitive:
Suppose you are rejected by one firm due to its algorithm's assessment of you. If different firms use different algorithms, this isn't a huge deal. The next firm will evaluate you with a different algorithm, and you might well get a job there. But if all firms use that same algorithm, then you'll likely be left jobless, the victim of `algorithmic blackballing,' to use \citet{ajunwa2020auditing}'s phrase.\footnote{As \citet[1]{kleinberg2021algorithmic} put it, `If all employers or lenders use the
same algorithm for their screening decisions, then particular applicants might find themselves locked out of the market when this shared algorithm doesn't like their application for some reason.' They do not explicitly endorse this, but describe it as a `common concern.'

\citet[see esp. ch. 6]{oneil2016weapons} advances a version of this argument. One of the risks of algorithmic decision-making, in her view, has to do with scale: `scale is what turns WMDs [weapons of math destruction] from local nuisances into tsunami forces, ones that define and delimit our lives' (p. 29). And she tells the story of Kyle Behm, whose job application was rejected over and over. He eventually learned that all the firms he applied to were using the same hiring algorithm, one which seemingly picked up on his mental health condition and downgraded him on that basis.}

Creel focuses on r\'{e}sum\'{e} screening algorithms, which screen out some applications and screen in others, which are then considered by human recruiters. Under monoculture, the same applications will be screened out by every firm, and in this sense, the corresponding candidates will be systemically excluded. By contrast, under polyculture, different applications will be screened out, and different ones screened in, by different firms. And so more candidates will have their applications considered by humans, with fewer thus being systemically excluded. 

\citet{CreelManuscript-CREAMA} then grounds the systemic exclusion objection in contractualism \citep{Scanlon1998-SCAWWO-5}. The idea is that those systemically excluded under monoculture could reasonably reject monoculture as an overall system, provided there exists a feasible, low-cost alternative that reduces or eliminates systemic exclusion (instead of just relocating it to a different subset of the population). Her preferred alternative involves randomization by individual firms, as opposed to a polyculture of deterministic algorithms. But in any case, the key is that there exist feasible, low-cost alternatives which reduce this sort of systemic exclusion. 

We are skeptical of Creel's emphasis on whether or not someone is screened in, over and above whether or not they ultimately get the job. She says that only those whose r\'{e}sum\'{e}s are screened in by a firm and considered by a human count as having had a chance at the job, or an opportunity to get it. She writes that `a chance is an algorithmic score that guarantees that the candidate will receive human consideration' (p. 23) and describes those screened out as having been `excluded from opportunity' (p. 24). This is mistaken. Consider sports which distinguish between a regular season and the playoffs. It would be absurd to say that those teams which don't wind up making the playoffs started the season with no chance to win the championship. Of course, once we get to the playoffs, those teams that didn't qualify no longer have a chance. But they still had a chance at the outset. Similarly with r\'{e}sum\'{e} screening and jobs. You may no longer have a chance once you fully submit your application and are screened out. But you still had a chance,\footnote{This is especially true if chances are epistemic rather than objective. We don't know which candidates will be screened out, and so everyone initially has a non-zero epistemic chance of being screened in. But everyone may \textit{also} have a non-zero objective chance of being screened in, at least prior to submitting their applications. See \citet[20-2]{CreelManuscript-CREAMA} for related discussion.} and certainly an opportunity, earlier in the process, for instance when developing your application materials or while engaging with interactive screening algorithms. 

Nonetheless, one might claim that being screened in benefits you, even holding fixed whether or not you ultimately get the job. We concede that in some cases, it could. If you learn that you were at least screened in, this might give you confidence and motivation. And if you thereby get an interview, this gives you a bit of practice, even if you don't get the job. One might even think that receiving human consideration is itself an intrinsic good. Creel is officially neutral here. But if being screened in is itself a good, perhaps it ought to be distributed more evenly, whether for Creel's contractualist reasons or ordinary consequentialist ones stemming from goods having diminishing marginal utility. 

But being screened in will often be of little or no benefit if you don't wind up getting the job, especially if you never learn you were screened in and don't get an interview. So a system that grants human consideration to more candidates won't be meaningfully better, unless it also changes who gets jobs. We therefore focus on how monoculture and polyculture affect which candidates ultimately get jobs. 

We also focus on a setting in which hiring is done wholly algorithmically, with algorithms providing a full ranking of candidates, as opposed to just screening some in for human review and others out.
This allows us to study a limiting case of monoculture, where firms' decisions are maximally correlated via an algorithm.
If this extreme form of monoculture is fairly benign, then less extreme forms are probably benign as well.\footnote{One can also imagine the whole hiring process, including any interviewing, being fully automated.} 

In our preferred setting, can we formulate a compelling objection to monoculture based on systemic exclusion? We'll consider several attempts to do so. We begin with the objection sketched at the beginning of this section. This objection says that under monoculture, if you're rejected by one firm, then you'll be rejected by the others, since they're all using the same algorithm. But this is mistaken. Suppose that the firms move in sequence, each hiring its top-ranked candidate from among those remaining. In monoculture, the firms use a shared algorithm, and so they rank candidates in the same way. Then, under monoculture, even if you don't get hired by one firm, you might well get hired by the next, in particular if you were the runner-up at the former. Rejection by one firm doesn't mean rejection by all, since the pools of remaining candidates aren't the same for all firms.

More generally, with a fixed number of jobs and a fixed number of job candidates, the same number of candidates will be left jobless under monoculture as under polyculture \citet{peng2024monoculture}, regardless of whether firms make offers sequentially or simultaneously.\footnote{As \citet[16]{peng2024monoculture} write, `Our work challenges the implicit assumption that the widespread adoption of a common algorithm necessarily increases the number of individuals excluded from all opportunities. When the number of “opportunities” (positions at a firm, seats at a college) is constant, the number of individuals that receive these opportunities (and who are denied from all opportunities) should remain constant.' We discuss their model in \S\ref{subsec:aggregation}.}

Of course, if we imagine that a firm will hire you if and only if its algorithm gives you a score above some threshold, one which is fixed for all firms and doesn't depend on whether we have monoculture or polyculture, then you are indeed more likely to be hired under polyculture than under monoculture.\footnote{This assumes that the scores given by the algorithm under monoculture are not systematically higher than those assigned by the various algorithms under polyculture.} But that's because we're imagining a situation in which there are more jobs under polyculture! Absent any empirical evidence, however, we should not expect polyculture to yield a substantially lower unemployment rate than monoculture.\footnote{A small caveat is that monoculture could result in congestion in the hiring market, where searches fail as firms herd around the same candidates. See \citet{baek2025hiring} for discussion of how firms under monoculture might act strategically to mitigate such congestion.}

Now, even if monoculture and polyculture will exclude the \emph{same number} of people from jobs, they will likely exclude \emph{different} people. This would be especially concerning if monoculture increased the probability of systemic exclusion of members of some socially salient group, like those of a certain race or gender. In this vein, \citet[1]{bommasani2022picking} worry that monoculture will `reinscribe social hierarchy.'\footnote{\citet[37]{creel2022algorithmic} also raise this concern, though they think that monoculture is problematic even if those excluded under it don't constitute a socially salient group. \citet{CreelManuscript-CREAMA} explicitly argues that harms from monoculture cannot be reduced to discrimination.}
This yields a different way of cashing out the objection from systemic exclusion: monoculture is problematic, since it raises the risk of systemically excluding members of some socially salient group.

Here's why you might think that monoculture is more likely to systemically exclude members of some socially salient group. Suppose that there are 100 candidates and 90 firms, which will each hire a single candidate. Then, under monoculture, members of a given group will all be left jobless if and only if the shared algorithm ranks them in the bottom 10\% of candidates. But under polyculture, they'll be left jobless if and only if each of the 90  algorithms ranks them in the bottom 10\%. And it's far more likely that one algorithm will be biased against that group than that all 90 will. 

But this is mistaken. It's true that under monoculture, you'll be jobless if and only if the shared algorithm ranks you in the bottom 10\%. And it's true that under polyculture, this will happen \textit{if} every algorithm ranks you in the bottom 10\%. But it's false that under polyculture, this will happen \textit{only if} they all rank you in the bottom 10\%. Indeed, you can be left jobless even if all algorithms rank you 2\textsuperscript{nd} overall; this will happen if each firm hires one person and each algorithm yields a different top-ranked candidate. So while it is indeed more likely that a single algorithm will be biased against some group than that many will be, this doesn't mean that monoculture is more likely than polyculture to systemically exclude its members, since there are many other ways for polyculture to yield such exclusion.\footnote{An example may help. Consider exclusion of individuals rather than groups. Let there be 3 candidates (A, B, and C) and two firms (Firm 1 and Firm 2, with the former hiring first). Under monoculture, there are six ways for the shared algorithm to rank the candidates, and candidate C is left jobless in two of these, namely those where they're ranked third (i.e., $A\succ B\succ C$ or $B\succ A\succ C$). 

Under polyculture, there are six ways for each of the two algorithms to rank the three candidates, so there are 36 possible pairs of rankings. Any given candidate is left jobless in 12 of these, namely all of those in which Firm 1's algorithm ranks someone else first, and Firm 2's algorithm ranks them below the other remaining candidate. So candidate C will be left jobless whenever (i) Firm 1's algorithm ranks A first (i.e., $A\succ_1 B\succ_1 C$ or $A\succ_1 C\succ_1 B$) and Firm 2's algorithm ranks B above C (i.e., $A\succ_2 B\succ_2 C$ or $B\succ_2 A\succ_2 C$ or $B\succ_2 C\succ_2 A$), or (ii) Firm 1's algorithm ranks B first (i.e., $B\succ_1 A\succ_1 C$ or $B\succ_1 C\succ_1 A$) and Firm 2's algorithm ranks A above C (i.e., $A\succ_2 B\succ_2 C$ or $B\succ_2 A\succ_2 C$ or $A\succ_2 C\succ_2 B$). There are six possibilities in (i) and six more in (ii), for 12 total.}

Of course, rebutting that tempting reasoning doesn't settle the question of
which system is more likely to systemically exclude some group. But the answer
to that question depends on assumptions about what information we have available
and how algorithms operate with respect to socially salient groups. Suppose we assume a certain veil of ignorance where we have no information about the
algorithms in question, so that algorithms can be treated as randomly generated rankings of candidates.  Then, under both monoculture and polyculture, every combination of $n$ people (where $n$ is the number of candidates minus the number of jobs) is equally likely to constitute those left jobless. And so any socially salient group is equally likely to be systemically excluded under monoculture as under polyculture.\footnote{To illustrate, suppose that our group of 100 candidates is partitioned into 10 socially salient groups of 10 candidates each. Then, behind this veil of ignorance, the probability that any given 10-candidate group (whether socially salient or not) is left jobless is one divided by the number of different ways of choosing 10 candidates out of 100, i.e., $\frac{1}{\binom{100}{10}} = \frac{10!\times 90!}{100!}$. This is true both under monoculture and under polyculture.}

Of course, relative to other states of information, monoculture is more likely to systemically exclude some socially salient group. Suppose we know that every  algorithm we might deploy has an extreme bias against some group or other, ranking all its members at the bottom. Then, monoculture is guaranteed to systemically exclude some socially salient group. But under polyculture, it's still possible  that someone from every group gets a job. 

Having said that, there are other states of information relative to which polyculture is more likely to systemically exclude some socially salient group. Suppose we know that \textit{no} algorithm we might employ has the sort of extreme bias just considered. Then, monoculture is guaranteed \textit{not} to systemically exclude any socially salient group. But polyculture could still do so, since it can systemically exclude some socially salient group even if its members aren't all ranked at the bottom by any algorithm.

We conclude that it is far from clear that monoculture increases the risk of systemic exclusion of socially salient groups. There are models where this risk is the same under monoculture and polyculture, models where it is higher under monoculture, and models where it is higher under polyculture. And none are obviously more realistic than the others.  

The above models do, however, illustrate one point in monoculture's favor. Monoculture might make it easier to detect and eliminate discrimination and other undesirable phenomena, since hiring outcomes are directly and transparently tied to a single algorithm's ranking. By contrast, under polyculture, outcomes result from a complex interplay between many different algorithms' rankings and other structural features like the firms' hiring order.\footnote{See \citet{kleinberg2020algorithms} for broader discussion of how algorithmic decision-making can make it easier to detect and mitigate bias.}

Let's change tack. Perhaps we can rehabilitate the exclusion objection by framing it in terms of \textit{chances} rather than \textit{outcomes}. Perhaps polyculture gives more people a \textit{chance} of getting a job. We have already mentioned chances in connection with r\'{e}sum\'{e} screening algorithms, where we rejected the thought that only those screened in had a chance of getting the job. What about in our setting, where the whole hiring process is done algorithmically?  Under monoculture, one might think, your chance of getting a job is either 0 or 1, depending on where you're ranked by the sole algorithm being used. But under polyculture, most people will have intermediate chances of getting a job. And so fewer people will have non-zero chances of getting a job under monoculture than under polyculture. Then, even if monoculture and polyculture are equally exclusionary from an \textit{ex post} perspective (since the same number will wind up jobless in either case), monoculture is more exclusionary than polyculture from an \textit{ex ante} perspective.

Now, an initial thing to say in response is that monoculture needn't be
deterministic, and polyculture needn't be stochastic. For the latter, note that if each algorithm is deterministic, and the order in which firms hire is settled, each candidate's chance of getting a job is either 0 or 1. Of course, in real life, there is likely to be some chanciness in the order in which firms hire, and in polyculture (but not monoculture), changing the hiring order can change which candidates get jobs. (Suppose Firm 1 ranks candidates $A\succ B\succ C$ and Firm 2 ranks them $A\succ C\succ B$. If Firm 1 goes first, then $A$ and $C$ get hired, with $B$ left jobless, while if Firm 2 goes first, then $A$ and $B$ get hired, with $C$ left jobless.) More importantly, monoculture need not be fixed or deterministic.\footnote{As noted, Creel favors randomization in r\'{e}sum\'{e} screening algorithms.} Firms regularly experiment with deployed algorithms, and algorithms can behave non-deterministically for a variety of reasons~\citep{Redstone2025}.\footnote{Moreover, a single algorithm may deliver different results to different firms based on a candidate's estimated fit for that particular firm, rather than a static estimate of their quality~\citep{ha2015personalized}. Admittedly, this amounts to a bit of a departure from the sort of monoculture we've been considering, in which each firm uses the same ranking of candidates to make its decisions.}

Let us temporarily set this aside for the sake of argument, and assume that monoculture would give every candidate an extremal chance (0 or 1) of being hired, while polyculture would give every candidate an intermediate chance. Why should this favor polyculture over monoculture? Arguably, all that matters is what happens \textit{ex post}, and \textit{ex ante} considerations have no independent moral significance. But this anti-\textit{ex ante} view is controversial, and we don't want to rely on it. Many theorists think that chances matter for fairness. For instance, \citet{Broome1991-BROV-4} holds that when people have equal claims to some good (say, because they have the same need or merit), then if it cannot be divided between them, fairness requires that they be given equal chances of receiving it. 

But the cases we're considering are not at all like this. Job candidates are not all equally deserving. Some are more qualified and productive than others, and so they have unequal claims to being hired on grounds of merit.\footnote{Some are also better off than others, and so they have unequal claims on grounds of need.}  Now, Broome does claim that even when people's claims are not exactly equal, we should still give people intermediate, albeit unequal, chances of receiving the good.\footnote{\citet{jain2024scarce} appeal to Broome's view to argue that scarce resource allocations made using machine learning tools should involve some randomization.} Instead of an equally-weighted lottery to distribute the good, we should run an unequally-weighted one, where each person's chance of winning is proportional to the strength of their claim.\footnote{Compare \citet[1]{jain2024algorithmic}, who suggest that `Formal views of equality of opportunity often contend that all decision subjects should have a chance at a positive outcome and that their likelihood of receiving that outcome accords with their merit or desert.'} But he gives no argument for this; he just says that he finds it plausible.\footnote{He writes, `We have agreed that fairness requires everyone to have an equal chance when their claims are exactly equal. Then it is implausible it should require some people to have no chance at all when their claims fall only a little below equality' (p. 99).} We find it far from obvious. Perhaps when people have unequally strong claims, fairness requires that the good(s) simply be given to those with the strongest claims.\footnote{You might object that this would be unfair in an iterated setting, since the good(s) would always be given to the same people. But if people have claims on grounds of need or interest (and not just on grounds of merit), then the same people wouldn't always have the strongest claims at each time, since receiving the good in one time period would typically lessen one's need for it in the next.}

Having said that, we do not even need to reject Broome's theory of fairness in order to defend monoculture. Suppose that fairness requires that candidates' chances of being hired be proportional to their merit, and that everyone has non-zero merit and hence ought to have a non-zero chance of being hired. Then, even if monoculture gives everyone extremal chances of being hired while polyculture gives everyone intermediate chances, it \textit{still} does not follow that polyculture is better than monoculture on grounds of fairness. For the distribution of chances of being hired yielded by monoculture could still be closer to the ideally fair distribution than is that yielded by polyculture. After all, a probability distribution with only extremal probabilities can be closer to one with only intermediate probabilities than is another which also only has intermediate probabilities.\footnote{For instance, $\langle 1, 0\rangle$ is closer\textemdash both intuitively and on most measures\textemdash to $\langle 0.9, 0.1\rangle$ than is $\langle 0.5, 0.5\rangle$.}
This is an instance of a broader lesson: the fact that the ideal state has some property doesn't mean that non-ideal states with that property are better than non-ideal states without that property \citep{lipsey1956general}.  Perhaps in the ideally just society there would be no soup kitchens, but that doesn't mean that we should close the soup kitchens in our highly non-ideal society!

Let us consider a final formulation of the exclusion objection. Even if monoculture and polyculture are equally exclusionary in a single round of hiring, perhaps monoculture will tend to result in the same people being left jobless year after year. If so, monoculture would yield constant employment for some and long-term unemployment for others, whereas polyculture would yield a more equal distribution of jobs over time.\footnote{See also \citet[24]{CreelManuscript-CREAMA} for discussion of monoculture and inequality over time.}

To take a toy model, suppose that jobs last for one year, after which everyone reapplies to all firms. In this setup, why might you think that  monoculture would result in the same candidates being left jobless year after year? The thought would have to be that algorithms never change their rankings of candidates, for instance because they base their rankings solely on immutable features. But if algorithms never change their rankings, wouldn't polyculture similarly result in the same candidates being hired each year? Not necessarily. As we saw earlier, under polyculture (but not monoculture), the order in which firms hire can affect who gets a job, and this is something that could change year-to-year, even if algorithms never change their rankings. Polyculture has an additional source of noise and variability relative to monoculture, namely the hiring order of the firms.\footnote{In the simultaneous hiring models considered in \S\ref{subsec:aggregation}, candidates' preferences over the firms play a similar role. Even if firms' algorithms never change their rankings of candidates, changes in candidates' preferences can result in changes in who gets hired.} 

Having said that, it's implausible that decent algorithms would base their rankings solely on immutable features or, more generally, never change their rankings of candidates. Any algorithm worth its salt will care about things like recent job performance, new credentials, and other features that can change over time. What happens when we allow algorithms to rank candidates differently in different years? Well, when algorithms change who is ranked at the bottom, then for monoculture, this is \textit{guaranteed} to change who is left jobless. But not so for polyculture. Under polyculture, the same people can be left jobless even when every algorithm completely reverses its ranking from one year to the next. To illustrate, suppose that in the first round, Firm 1 ranks candidates $A\succ B\succ C$ while Firm 2 ranks them $C\succ B\succ A$. Here, $A$ and $C$ are the ones who get hired, and $B$ is left jobless. Then, in the next year, each firm's algorithm completely reverses its ranking, as a result of changes in the candidates' features. So now, Firm 1 ranks them $C\succ B\succ A$ while Firm 2 ranks them $A\succ B\succ C$. And again, $A$ and $C$ are the ones who get hired, with $B$ left jobless. Similarly, under polyculture, the same candidates can be left jobless even when they all move up in every algorithm's ranking from one year to the next. To illustrate, suppose that in the first year, Firm 1 ranks candidates $A\succ B\succ C$ and Firm 2 ranks candidates $B\succ A\succ C$, and so $A$ and $B$ are hired, with $C$ left jobless. In the second year, Firm 1 ranks candidates $A\succ C\succ B$ and Firm 2 ranks candidates $B\succ C\succ A$. Then, even though $C$ has moved up in each firm's ranking, once again $A$ and $B$ are hired, with $C$ left jobless. So monoculture may be more responsive to changes in rankings than is polyculture.  

Stepping back, inequality in employment over time will depend on how much
across-period variance there is in the selection
process~\citep[e.g.,][]{shorrocks1978measurement}. With little or no variance, the employment distribution will be highly unequal, whereas with high variance, we should expect more equality over time. The potential argument, then, is that monoculture will yield less variance than polyculture. This is certainly possible. But there are also reasons to believe that polyculture could have lower variance. As we discuss in \S\ref{sec:information}, a key potential benefit of polyculture is its ability to aggregate information via the ``wisdom of the crowd.'' But aggregating information to increase accuracy typically reduces variance! Intuitively, a system that produces dramatically different predictions from one year to another cannot be accurate (unless the world itself has changed dramatically in that time).

We conclude that it's far from clear that polyculture will yield greater
variability in who is left jobless from one year to the next. Even if
monoculture may appear to reduce variance by applying a single, fixed evaluation
scheme, the analogous polyculture need not be any more dynamic. A fixed system,
collectively aggregating information from heterogeneous algorithms, may still
arrive at the same decisions year after year. And to the extent that
polyculture's wisdom of the crowd aggregates information more efficiently than a
monoculture, we might instead expect to see \textit{greater} inequality under
polyculture.

We have considered a range of attempts to cash out the thought that monoculture is objectionable as a result of increasing the risk of systemic exclusion. And we have argued that none are compelling. So at least as things stand, we conclude that considerations of exclusion do not favor polyculture over monoculture.

\section{Agency}
\label{sec:agency}

Perhaps monoculture is objectionable for reasons having to do with
\textit{agency}. First, one might worry that monoculture gives people
\textit{too little agency}, limiting their ability to improve their outcomes through hard work, discipline, and ingenuity. Second, and coming from the opposite direction, one might worry that monoculture gives people \textit{too much} agency, in some sense, by over-incentivizing strategic behavior and `gaming.' We have not seen these worries advanced in the literature, but versions of them have come up often in conversation. 

Start with the first. It is important for people to be able to exercise their agency and work for better outcomes.\footnote{As \citet{jain2024algorithmic} put it,  ``having a plurality of opportunities \dots is a prerequisite for the accompanying virtues of freedom and the ability to shape our lives \dots'' Note that while they also oppose algorithmic monoculture, they do not explicitly advance the objection we consider here.} But, the worry goes, monoculture may limit this ability. After all, suppose that you apply to a job and get rejected. Polyculture would let you learn from that rejection, improve your application materials, and practice your interview skills, thereby improving your chances of getting a job at the next firm to which you apply. But a very strict version of monoculture rules this out. If your scores on one firm's application process are directly forwarded to another firm, or if you are scored once and for all by a centralized platform, then you cannot hope to improve your future outcomes through your own effort. 

This objection is unconvincing. For one thing, it targets a caricatured  monoculture whose lone algorithm bases its rankings only on immutable features of candidates. A more attractive monoculture would have its  algorithm be sensitive to features like education which candidates can change through effort. For another, some ways in which polyculture gives people more agency may be undesirable. For instance, in the context of matching markets, \citet{peng2024monoculture} observe that submitting more applications will tend to improve one's match quality under polyculture but not under monoculture. But those who can submit more applications tend to be wealthier. More generally, those who are already better-off tend to have a greater ability to invest time, effort, and resources to improve their outcomes.\footnote{This needn't always be the case. As \citet{bambauer2018algorithm} note, sometimes it is the disadvantaged who will benefit from a greater ability to improve outcomes through time and effort. For instance, those who are unemployed might have more time available as well as greater motivation.}

Second, and coming from the opposite direction, one might worry that monoculture gives people \textit{too much} agency. In particular, it might incentivize `gaming,' or  improving your score on some metric without improving your underlying quality. Goodhart's Law says that when a measure becomes a target, it ceases to be a good measure. And with monoculture, there's one big target that everyone is aiming at. So perhaps monoculture incentivizes `gaming' to a greater extent than does polyculture.  

\citet{kleinberg2020classifiers} present a model under which candidates invest effort strategically to maximize their performance along different axes. There, gaming is most effective when there is a single evaluation for the candidate to target. Intuitively, gaming some metric increases your score on that metric by a lot, but leaves your scores on other metrics unchanged, whereas genuinely improving your quality increases your score on every metric, albeit by a more modest amount. This model might suggest that monoculture would incentivize gaming more strongly than polyculture, since there's just one algorithm in the former but many in the latter.

However, polyculture might still incentivize gaming. It might be better to pick one firm (or a few) and game its algorithm rather than genuinely improving your quality. The former could significantly improve your chances at that chosen firm, whereas the latter could improve your chances at every firm, but only modestly. We might also disincentivize gaming by converting our polyculture into a monoculture through `ensembling,' packaging the many algorithms together into a single one which, say, assigns each candidate the average score they received from the many. Then, genuine improvement will be better than gaming, provided that it's more effective at improving this average score. 

We have seen that the relationship between monoculture and strategic behavior is nuanced. In the case of \citet{peng2024monoculture}, candidates act strategically to \textit{submit more applications}, and polyculture leads to greater returns on their efforts. Coming from the other direction, the model of \citet{kleinberg2020classifiers} suggests that a simple-minded monoculture might be easily gameable. But a more sophisticated monoculture might effectively disincentivize gaming. Thus, the relationship between monoculture and strategic behavior is complex, and cannot be reduced to a single directional statement.

Monoculture's impact on agency is further complicated by the varying objectives
that decision subjects might have. Consider a version of polyculture in the
college admissions process in which half of all colleges accept SAT scores and
the other half accept ACT scores. How should a student prepare? For some
students, it may be prohibitively difficult to study for both exams. Instead, they might choose to focus on the SAT and only apply to colleges that accept
SAT scores. The student exercises agency to select which exam to take, how much
effort to invest in preparation, and which colleges to apply to.

Now consider what happens when the system shifts to a monoculture, where all
colleges accept both test scores. How does this impact a student's agency? Right
away, it might appear that students have \textit{more} agency: they can choose
the exam they feel best suited for while still applying to all of the colleges
of their choice. But not every student necessarily experiences increased agency.
Under the previous polyculture, an enterprising student could try to prepare for
\textit{both} exams, allowing them to apply to two disjoint pools of colleges
and increasing their overall chances of admission. Polyculture effectively
limits the competition for each pool of colleges, allowing a highly motivated
(or highly-resourced) student to derive benefits from additional effort. This is
no longer effective under monoculture, since every student is now competing
against the entire student population for each college.
Each student simply chooses an exam (ACT or SAT) and focuses their preparation
there.\footnote{Assuming students submit the max score between the two exams on a percentile basis.}

Structurally, polyculture and monoculture create different incentives.
Polyculture allows students to target a particular pool of colleges, based on
both their preferences and the overall competitive landscape. Monoculture breaks
down barriers to put students in the same competitive pool. For some students,
this increases their ability to better their outcomes, whereas others may find
that their efforts are no longer effective. By removing the ability to target
niche markets or exploit information silos, monoculture forces all candidates
into a single, transparent competition. In this way, the shift to monoculture
does not so much reduce agency as it does redistribute its returns and change
its primary mode of expression.

Ultimately, the relationship between monoculture and agency is far from
straightforward. While critics worry that a single decision-making tool will rob
individuals of the ability to improve their prospects, we have seen that the
same features of monoculture can, in other contexts, make strategic improvement
easier or reduce the unfair advantages enjoyed by those with the resources to
navigate a fragmented polyculture. Overall, how monoculture affects agency, and whether these effects are good or bad, will depend on the specific setting in question. But in our view, considerations of agency do not provide a compelling reason to prefer polyculture over monoculture, though we expect future research to offer more insights on the impacts of monoculture on agency in practice.

\section{Information}
\label{sec:information}

The final sort of objection to monoculture that we'll consider says that it is suboptimal from an informational perspective. We consider two versions of this objection. The first says that polyculture is likely to yield more accurate judgments at any given time, since it takes advantage of the `wisdom of crowds.' And the second says that polyculture will yield more exploration, generating information that allows it to outperform monoculture over time. 

\subsection{Aggregation and the Wisdom of Crowds}
\label{subsec:aggregation}

Start with the first. Again, this objection says that polyculture will probably be more accurate than monoculture, which in the hiring case means that it makes it more likely that objectively better candidates will be hired. And that's because polyculture takes advantage of the `wisdom of crowds,' whereby a group of diverse decision-makers does better in the aggregate than any single one. 

Two recent papers give models suggesting an argument along these lines. \citet{kleinberg2021algorithmic}  give a model in which two firms move in a random order, each hiring a single candidate. There is some objective ranking of candidates, to which firms have only imperfect access.\footnote{The assumption of a single objective ranking of candidates is, of course, an idealizing assumption.} Different algorithms have different accuracy levels, understood roughly as probabilities of yielding or approximating the objective ranking. But all algorithms under consideration are better than random. 

Each firm chooses whether to use its own, less accurate, private algorithm or a single, more accurate, public algorithm. Their private algorithms are independent of each other.\footnote{That is, conditional on the true ranking, one private algorithm's yielding a given ranking is probabilistically independent of the other's doing so.} Then, against some plausible assumptions, Kleinberg and Raghavan show that for any accuracy level for the private algorithms, there is some higher accuracy level for the public algorithm such that (i) each firm will prefer to use the public algorithm regardless of what the other firm does, but (ii) the expected average quality of the two candidates hired is higher if both use their private algorithms than if both use the public one.

For our purposes, the latter result is key. It says that polyculture can be
better than monoculture overall, even if the public algorithm is better than the
various private algorithms and firms make rational choices. Intuitively, this is
because independent decision-makers can collectively bring more information to
the table than a centralized decision-maker.

\citet{peng2024monoculture} give a different model with a similar upshot. There are many firms and many candidates. Candidates have idiosyncratic (here, uniformly random) preferences over firms. Each firm has the same number of jobs, and the total number of jobs is less than the total number of candidates. Firms want all their jobs filled, and every candidate prefers having a job to being unemployed.

Peng and Garg use the concept of a stable matching \citep{gale1962college}, which is an allocation of
candidates to firms such that no firm-candidate pair `would jointly prefer to
defect from the current matching to be matched with each other' (p. 4). They show that in their model, a matching is stable if and only if it is characterized by shared cutoffs across all firms, with each firm extending offers to all candidates to whom its algorithm gives a score above its cutoff and candidates accepting their most preferred offer.\footnote{They assume a continuum of candidates, from which it follows that stable matchings are characterized by a vector of firm-specific cutoffs which are market-clearing, meaning that they result in all jobs being filled \citep{azevedo2016supply}. In Peng and Garg's model, because firms are symmetric and candidates' preferences over firms are uniformly random, the cutoff for each firm is the same.}

Under monoculture, each firm uses the same algorithm, while under polyculture, each firm uses a different one. Peng and Garg assume that the available algorithms are independent but equally accurate in expectation. Specifically, each algorithm's score for each candidate is equal to their objective value plus some noise value, with each noise value being drawn independently from some fixed noise distribution.

Peng and Garg show (under certain assumptions) that under polyculture, as the
number of firms increases, the probability that only the best candidates are
hired goes to 1.\footnote{This holds as long as the noise distribution has a
  lighter-than-exponential tail. \citet{peng2024wisdom} show that the opposite
  holds for heavy-tailed noise distributions, meaning polyculture amounts to
  uniform random guessing in the limit. They further extend their results to the
  case where firms have different numbers of jobs, and where candidates have
preferences over firms that needn't be uniformly random.} And so in the limit, polyculture is more accurate than any given monoculture, provided the latter is still imperfect. This is because under polyculture, a candidate gets an offer if and only if the highest score they get from any algorithm exceeds the shared cutoff.\footnote{By contrast, under monoculture, a candidate gets an offer if and only if the score they get from the one algorithm exceeds the cutoff. The shared cutoff will therefore be higher under polyculture than monoculture.} And if one candidate has a higher value than another, the probability that the former's highest score is greater than the latter's goes to 1 as the number of algorithms goes to infinity.\footnote{This assumes that the noise distribution has a lighter-than-exponential tail. \citet{peng2024wisdom} prove that if it is long-tailed, then we get the opposite result,  namely that under polyculture, all candidates have the same probability of getting a job, regardless of their objective values.}\textsuperscript{,}\footnote{This result is reminiscent of the Condorcet Jury Theorem. Suppose that everyone in a group has the same better than random chance of getting the right answer to a given yes-no question. And suppose that any one member getting the right answer is probabilistically independent of any other member doing so. Then, as the size of a group increases, the probability that its majority judgment is right goes to 1.}

The foregoing results sketch mechanisms by which polyculture \textit{might} outperform monoculture, but they do not conclusively establish that it \textit{will} do so. 
We raise three objections that limit the force of this wisdom-of-crowds-based
argument against monoculture.\footnote{The three objections we raise directly
  speak to the implications of monoculture for social welfare in aggregate. A
  separate objection relates to \textit{transfers} of welfare due to
  monoculture: even if polyculture is overall better than monoculture for firms,
  monoculture may be better for successful job candidates. Suppose that firms
  make offers simultaneously. Then, under monoculture, if a candidate gets an
  offer from one firm, they get an offer from all.  This means that any
  candidate who winds up with a job winds up with a job at their first-choice
  firm \citep[Thm. 2]{peng2024monoculture}. But this is not the case under
  polyculture. 

This echoes observations made about mechanisms for matching students with schools, given students' expressed preferences and schools' capacity constraints. In single tie-breaking (STB) systems, each school uses the same (randomly generated) ranking of students, while in multiple tie-breaking (MTB) systems, each school uses an independent randomly generated ranking of students. STB is analogous to monoculture, while MTB is analogous to polyculture. As \citet{peng2024monoculture} note, empirical and theoretical work on school choice indicates that STB results in more students being matched to one of their top choice schools than does MTB. See e.g., \citet{ashlagi2019assigning}. 

Monoculture's concentration of offers on fewer candidates may also enable successful candidates to bid up their wages. Having multiple offers gives you greater bargaining power, since it can be costly for the firm to find a replacement if you decline. (Multiple offers usually also give firms additional `higher-order' evidence of your quality, but this won't happen if all firms know that they are operating in a monoculture.) Having said that, most recent economics literature on the effects of multiple offers on bargaining concerns `on the job search,' or incumbent employees receiving outside offers \citep{cahuc2006wage}; it is unclear how much bargaining power unemployed workers get from having multiple offers.}
First, both models considered above assume that under polyculture, decision-makers make independent errors. This is implausible, regardless of whether polyculture involves human decision-makers or algorithmic ones. Research has shown that humans make
correlated errors, due to shared culture, systematic biases, and overlapping information \citep{tversky1982judgment}. For instance, hiring managers tend to overestimate the quality of candidates who are physically attractive \citep{beauty-labor-market} and to underestimate that of candidates who speak with a non-native accent~\citep{huang2013political,taveras2025foreign}.
Algorithms are also likely to make correlated errors, in part because they tend to be trained on overlapping datasets and have access to similar pieces of information about candidates, such as r\'{e}sum\'{e}s and  LinkedIn profiles \citep{bommasani2022picking}. Increasing the  correlation between algorithms will reduce the extent to which polyculture can take advantage of wisdom of crowds-type mechanisms and thereby outperform monoculture.\footnote{Compare the controversy around the analogous independence assumption that figures in the Condorcet Jury Theorem. For discussion, see \citet[ch. 5]{Goodin2018-GOOAET}.}

Second, the gap between monoculture and polyculture narrows significantly when we consider market forces.
\citet{peng2024monoculture} show that polyculture can be more accurate than monoculture in identifying the best candidates.
In their model, under monoculture, all firms are forced to use the same ranking over candidates.
Suppose, instead, that firms could decide for themselves whether to follow a shared ranking or create an independent one for themselves.
Would they still choose the shared ranking?
More generally, does monoculture perform worse than polyculture \textit{at equilibrium}?
As noted, \citet{kleinberg2021algorithmic} show that this is indeed possible: they construct scenarios where the algorithm available under monoculture is marginally more accurate than any \textit{individual} independent decision-maker, but leads to worse social outcomes if multiple firms use it.
Importantly, they show that their construction is an equilibrium, meaning that each firm rationally chooses to adopt the shared algorithm, even though it makes them collectively worse off---a form of Braess's paradox, related to the classic prisoner's dilemma.
Unlike the strong separation shown by \citet{peng2024monoculture}, \citet{kleinberg2021algorithmic} only show a modest gap in performance  (understood as the sum or average of the objective values of the candidates who get hired) between monoculture and polyculture, and only under particular conditions.

A generalization of \citet{kleinbergprice} provides a tight characterization of
just how bad monoculture can be when firms behave strategically. In their
setting, firms choose from a menu of ranking technologies with varying quality.
Some are independent, whereas others are shared. They show that the price of
anarchy (i.e., the multiplicative gap between equilibrium and optimal social
welfare \citep{koutsoupias1999worst}) is at most two; outside of
pathological instances, it can be much smaller than this. In other words, while
\citet{peng2024monoculture} suggest that monoculture can be arbitrarily worse
than polyculture from a social welfare perspective, if firms can
rationally choose whether or not to adopt a common technology, \textit{the
option to engage in monoculture} can only have modestly negative welfare. If firms benefit by seeking out independent sources of information, they will do so.

Third and finally, nothing prevents monoculture from being just as accurate, or even more accurate, than polyculture.
In principle, a single monoculture algorithm could integrate all of the information available to different polyculture algorithms, combining to form an ensemble that matches or even beats the performance of the overall polyculture system. For instance, in the model from \citet{peng2024monoculture}, we can take a given polyculture and turn it into an ensemble-based monoculture whose sole algorithm assigns each candidate the highest score that they receive from any of the different polyculture algorithms with which we started. Then, it follows immediately from their Theorem 1 that as the number of algorithms used to build the monoculture's ensemble algorithm increases, the probability that the ensemble-based monoculture hires only the best candidates approaches 1.

Even with a finite number of algorithms, an optimal ensemble-based
monoculture will outperform polyculture. Trivially, monoculture can match the
performance of polyculture simply by simulating it and exactly matching the
outcomes it yields. Monoculture might even perform better by aggregating
information more efficiently than the market does under polyculture. In the
model of \citet{peng2024monoculture}, polyculture effectively aggregates
information by using the maximum score that a candidate receives across all
firms to determine whether or not they receive a job. Efficient inference
depends on the underlying noise distribution, but generally, there are
lower-variance ways to aggregate multiple noisy
observations~\citep{fisher1925theory}.\footnote{For example, for observations with additive Gaussian noise, the optimal aggregator is the mean, not the max. For heavy-tailed noise distributions, the max provides no signal in the limit \citep{peng2024wisdom}.}

We confirm this intuition with simulations in both a sequential hiring setting like that of \citet{kleinberg2021algorithmic} and
a simultaneous hiring setting like that of \citet{peng2024monoculture}. In the former, firms move in a random order,
each hiring their top-ranked candidate from among those remaining. In the
latter, firms have 10 jobs each, candidates have random preferences over firms,
and candidates are stably matched to firms via the deferred acceptance
algorithm \citep{gale1962college}.\footnote{We use the version of the algorithm
in which candidates propose to firms, not \textit{vice versa}.}

In each run of the simulation, there are 1,000 candidates,  each with an objective value drawn independently from the standard Gaussian (or normal) distribution with mean 0 and standard deviation 1. Firms have preferences over candidates based on their noisy estimates of candidates' objective values. 

In (ordinary) monoculture, firms have a shared estimated value for each candidate, which is the candidate's objective value plus a single noise value for that candidate drawn from the noise distribution (Gaussian with mean 0 and standard deviation 0.5). In polyculture, firms have independent estimated values for each candidate, which are the candidate's objective value plus firm-specific noise values for that candidate, drawn independently from the noise distribution. Finally, in ensemble monoculture, firms have a shared estimated value for each candidate, which is the mean of the polyculture firms' estimated values for that candidate. 

We measured the performance of each hiring system in terms of the average
objective value of the candidates who are hired, normalized so that the
possible performance values fall between 0 and 1.\footnote{Specifically, let
Actual be the average objective value of the $n$ candidates hired under a given
system, Worst the average objective value of the $n$ worst candidates, and Best
the average objective value of the $n$ best candidates. Normalized performance
is the ratio (Actual - Worst)/(Best - Worst). So normalized performance is 0 if
the worst candidates are actually hired and 1 if the best candidates are
actually hired.} The results, for varying numbers of firms, averaged across
1,000 simulation runs, are shown in \Cref{fig:sequential,fig:simultaneous}. Standard errors are too small to show.
In all cases, the ensemble monoculture performed best, followed closely by
polyculture and somewhat further back by ordinary monoculture.
Intuitively, this is because the firms collectively had far more information
(more independent estimates of candidates' values) in polyculture and ensemble
monoculture than in ordinary monoculture. But in ensemble monoculture, this
information was pooled and all of it deployed in every hiring decision, whereas
in polyculture, only a small part of this collective information was deployed in
each decision.\footnote{Code and documentation for these and all other
simulations are available at \url{https://github.com/mraghavan/monoculture-bandits}.}

\begin{figure}[H]
    \centering
    \includegraphics[width=0.8\linewidth]{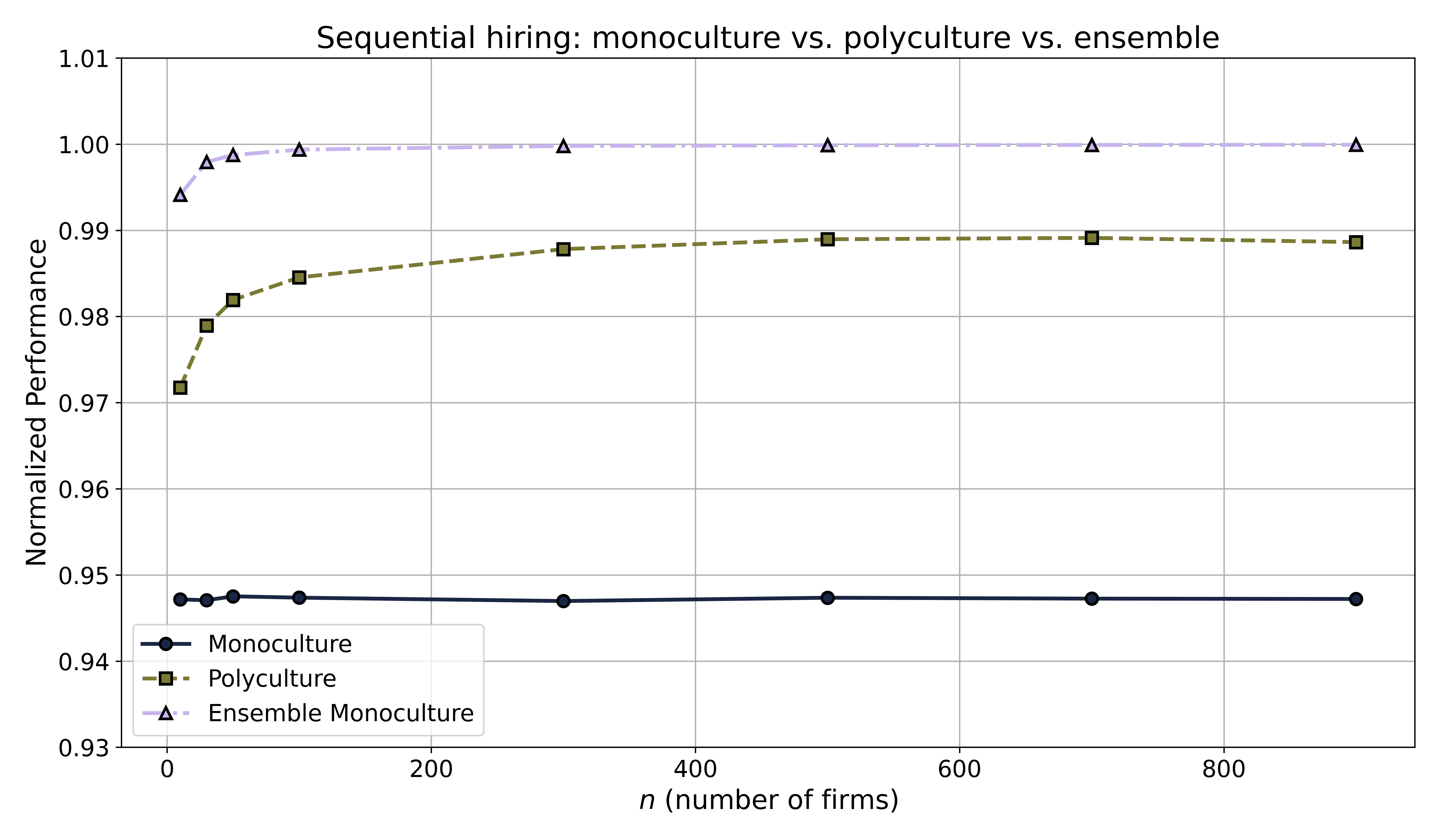}
    \caption{Performance in Sequential Hiring Setting.}
    \label{fig:sequential}
\end{figure}

\begin{figure}[H]
    \centering
    \includegraphics[width=0.8\linewidth]{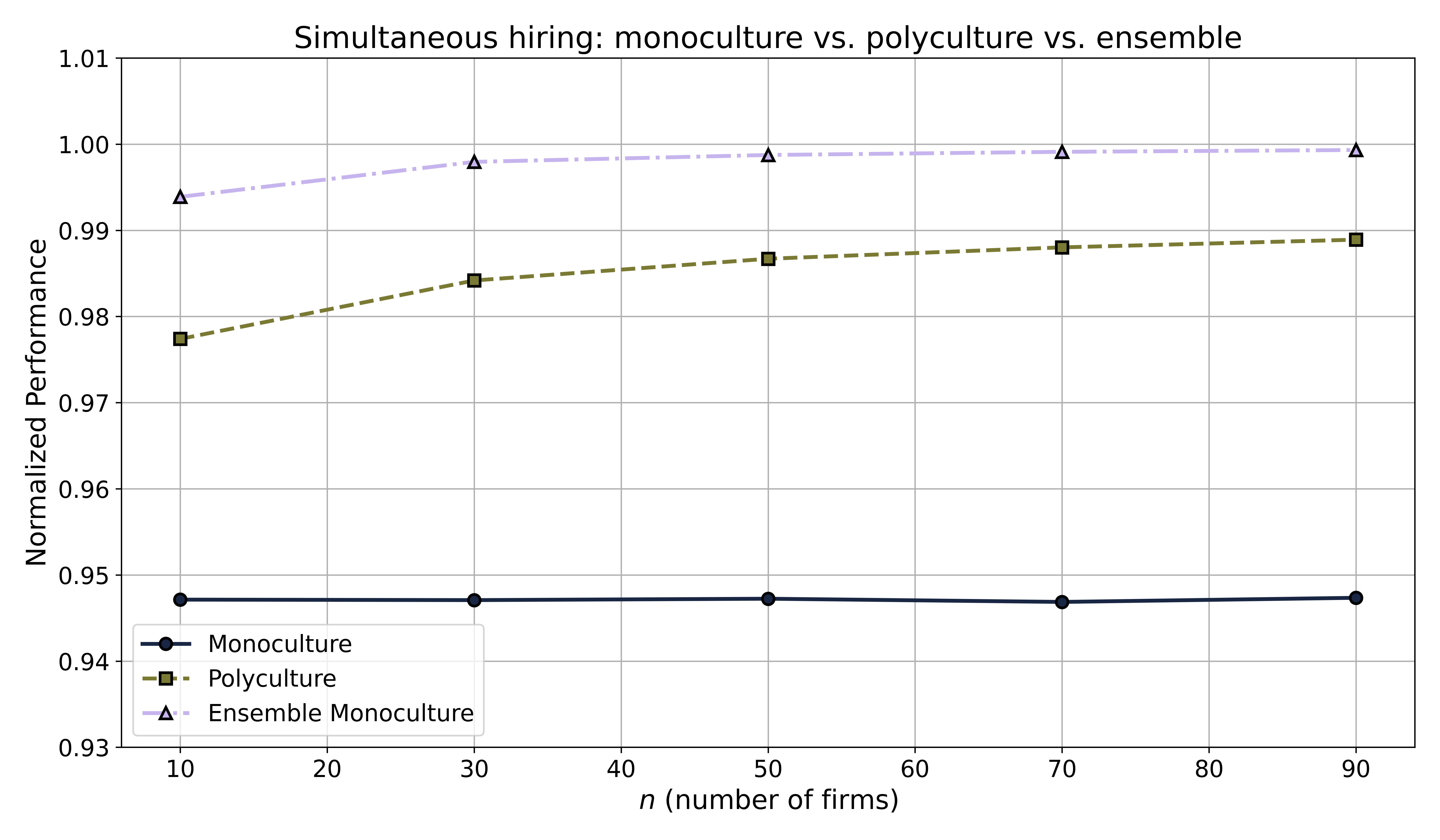}
    \caption{Performance in Simultaneous Hiring Setting.}
    \label{fig:simultaneous}
\end{figure}

One might question the relevance of the possibility of ensembling: Even though a monoculture based on a single ensemble algorithm could match or beat the performance of a given polyculture, doing so may require us to reproduce the work of each individual polyculture decision-maker. Collecting data, cleaning it, and building models is expensive.
Suppose that each polyculture algorithm uses overlapping but not identical input features. All algorithms may use a candidate's CV, but in addition, algorithm 1 might ask the candidate to take a personality assessment, while algorithm 2 scrapes the candidate's social media presence. Each of these may only be marginally useful relative to the other, and it may not be worth the effort for a central algorithm designer to incorporate both.
So even though monoculture could in principle match or improve upon the behavior of polyculture, is it unlikely to do so in practice? 

Again, market forces should mitigate the impacts of monoculture here.
In competitive markets, firms can benefit by bringing new information to the table.
If all of a firm's competitors rely on a single algorithm that only has access to a small set of information, a firm can identify high-quality candidates who are overlooked by the monoculture by seeking out new information.
The results of \citet{kleinbergprice} support this intuition: As long
as firms can rationally choose whether or not to solicit new information,
societal outcomes cannot be dramatically worse under monoculture. And any
monoculture that is stable must be accurate enough to preclude upstart
competitors from being able to gain a significant edge via novel algorithms.

\subsection{Exploration}
\label{subsec:exploration}

A final objection to monoculture, again relating to information, is that it threatens to result in a failure to explore. While this objection has not been developed in the literature,\footnote{It is very briefly mentioned by \citet{jain2024algorithmic} in their defense of `algorithmic pluralism.' They write that `many algorithms will undervalue applicants from under-represented groups and fail to learn about changes in applicant hiring potential over time. This should incentivize employers to not strictly follow algorithmic rankings and balance exploitation (selecting from groups with proven track records) with exploration (selecting from other groups to learn about quality.' See also \cite{Kitcher1990-KITTDO} and \citet{hong2004groups} for arguments that cognitive diversity improves exploration in science and other domains.} we think it is perhaps the most compelling. Algorithms often evaluate and choose between items with unknown values. For example, a hiring platform like LinkedIn must rank candidates in response to a search, based on predictions about whether each candidate is interesting to recruiters. Similarly, platforms like Netflix rank movies and TV shows based on estimates of their quality. But platforms also make mistakes. If the first few recruiters happen to scroll past a candidate's profile without engaging, the platform may take this as a signal that the candidate is of low quality and stop showing them in the search results. In an effort to return the best candidates for each search, platforms will ``exploit'' their knowledge of other high-quality candidates without ``exploring'' the candidates for whom their initial estimates may be mistaken. As a result, there may be many high-quality candidates whom the platform never gives enough of a chance to show their value, leading to an overall social failure to discover good candidates.

 We might imagine that a system without centralized information could avoid this
 pitfall. If recruiters relied more heavily on career fairs and conferences,
 they might be more likely to come across different candidates. Of course, such
 a system may be inefficient for a variety of reasons---each individual actor
 would only have a fraction of the overall information available to them. But it
 is precisely this lack of information that would lead them to experiment and
 discover unknown or unproven candidates.

How does this bear on the debate over monoculture and polyculture? It is
tempting to think that polyculture will yield more exploration than monoculture,
since under monoculture, an option which is ignored by one algorithm will be
ignored \textit{simpliciter}, whereas under polyculture, an option which is
ignored by one algorithm may wind up being chosen by another. The objection to
monoculture, then, is that it threatens to yield a suboptimal balance between
exploitation (doing what looks best given present knowledge) and exploration
(trying new things which might be even better). In particular, it does too much
of the former and too little of the latter. 

To assess this objection, we use the setting of multi-armed
bandit problems, which are standardly employed in computer science to
model and investigate trade-offs between exploration and exploitation (for an
overview, see \citet{slivkins2019introduction}). The name comes from `one-armed
bandit,' a euphemism for slot machines. In a multi-armed bandit problem, there
are $k$ `arms,' each with some unknown reward distribution. At each time $t=1,
2, \dots , T$, an agent selects one arm to pull and gets some reward. Crucially,
the agent only observes the reward of the arm selected; they get no information
about what rewards the other arms would have yielded. To maximize total reward
over time, agents must balance exploitation---pulling the arm with highest
expected reward, given the information available up until then---with exploration---pulling arms which are more uncertain and might be even better.

Optimal algorithms for bandit problems strike a balance between exploration and exploitation. One example is the Upper Confidence Bound (UCB)
algorithm~\citep{lai1985asymptotically, agrawal1995sample, auer2002finite}, which ranks
arms by an upper confidence bound on their expected rewards. The UCB algorithm thereby embodies the principle of `optimism in
the face of uncertainty,' giving a boost to arms about which there is greater
uncertainty.\footnote{Of particular relevance to our concerns here,
\citet{li2020hiring} model hiring decisions as bandit problems and compare
different algorithms using interview and hiring data from a Fortune 500 company.
They found that a UCB algorithm selected a more demographically diverse set of
candidates than did various ``greedy'' algorithms, which simply ranked
candidates by their estimated quality, intuitively because there was greater
uncertainty about candidates from minority backgrounds.}

From this work, we know that a monoculture \textit{could} engage in plenty of
exploration, namely if its sole algorithm was UCB or a similar algorithm. This
is an instance of a more general point: In theory, monoculture can generate
whatever result we want, provided we choose the right algorithm. But in
practice, there is no social planner who can just `choose the right algorithm.'
Instead, we must consider how individuals or firms, with their own goals and incentives, will
interact with a monoculture.

We worry that a monoculture deploying a UCB-style algorithm might be infeasible,
since such an algorithm might not be aligned with users' interests and
incentives. After all, algorithms that explicitly balance exploration and
exploitation must sacrifice expected utility today in order to learn, and
thereby better serve users tomorrow. For this reason, we consider what happens
when all agents choose \textit{greedily} or \textit{myopically}, pulling the arm with highest
instantaneous expected utility given their information, without any attempt to
explicitly explore. A platform that shows a recruiter the ``best'' available
candidates, according to its current knowledge, may be appropriately modeled as
greedy. Similarly, an agent acting on a user's behalf to find the
``best'' Mexican restaurant or mid-sized SUV should be understood as acting
greedily: No-one wants to purchase a car that is likely suboptimal in order to
gather information that might help future customers.

We already know that such greedy or myopic choice can cause a collective failure to learn. If some algorithmic system just seeks to maximize each  user's instantaneous expected utility (e.g.,~by recommending the most popular restaurants, products, or music), it will typically fail to explore and yield worse outcomes and information in the
long run~\citep{auer2002finite}.\footnote{As we will discuss later, most modern
  platforms do engage in \textit{some} amount of exploration, but optimal
explore-exploit algorithms are not commonly used in practice.} This qualitative
concern---algorithmic monoculture leading to under-exploration---arises
frequently in
the literature across news and journalism~\citep{nechushtai2019kind},
recommender systems~\citep{smets2022serendipity}, and
culture~\citep{chayka2025filterworld}.  While this is true regardless of whether the algorithm(s) in question form a polyculture or a monoculture, we consider whether  monoculture might \textit{exacerbate} the extent to which greedy behavior leads to societal exploration failures. Polyculture might outperform monoculture simply because independent algorithms choose differently, yielding greater collective exploration as a byproduct.

We can now describe the difference between monoculture and polyculture in the
context of multi-armed bandits. We model monoculture in terms of shared information and polyculture in terms of information siloes \citep[e.g.,][]{salganik2006experimental} which progress independently. Specifically, in monoculture, the agent who arrives at time $t$ observes the actions and rewards of all agents at times
$\tau < t$. In polyculture, the agents are partitioned into $k$ groups (corresponding to $k$ independent algorithms), and the agent at time $t$ only observes actions and rewards for prior individuals in the same group (i.e., using the same algorithm).

The objective we consider is societal information production, which we
operationalize as whether society collectively produces enough information to
correctly identify the best arm. In standard multi-armed bandits with greedy or myopic
individuals, there is some non-zero probability of such a failure
\citep{banihashem2023bandit}. Intuitively, if the best arm has an initial
``unlucky'' run of low rewards, greedy agents will have no incentive to
continue to explore it.

We formalize this setup as follows.
At each $t = 1, 2, \dots, T$, an individual arrives and must select between arms
1 and 2, each of which yields rewards from a Bernoulli distribution with
(unknown) parameters $\mu_1, \mu_2$ respectively.
Following \citet{banihashem2023bandit}, we assume a completely frequentist setting, where everyone observes a common initial sample of $N_0$ reward samples from each arm before $t=1$.
Similar results hold in a Bayesian setting where everyone has a common Beta prior.
Each individual chooses greedily, maximizing frequentist expected utility; that is, they simply choose the arm with the highest observed average reward. As noted above, the only distinction between monoculture and polyculture is the information structure: Under monoculture, each individual observes the choices and rewards from all prior individuals, whereas under polyculture, we have $k$ independent sets of individuals with no shared information across them (other than the $N_0$ initial samples that everyone observes).

Finally, we define failure to learn as follows.
After $T$ individuals have acted, the total information produced by society is just the initial $N_0$ samples per arm plus the $T$ observed rewards.
Critically, the distribution of those $T$ rewards across the two arms depends on the rewards themselves.
It is possible for only a small number of them to be from one arm or the other.
Our definition of success is simple: Do we correctly identify the best arm?
That is, if $\mu_1 > \mu_2$, does arm 1 perform better than arm 2 in the observed sample?

\begin{figure}[ht]
    \centering
    \includegraphics[width=0.8\linewidth]{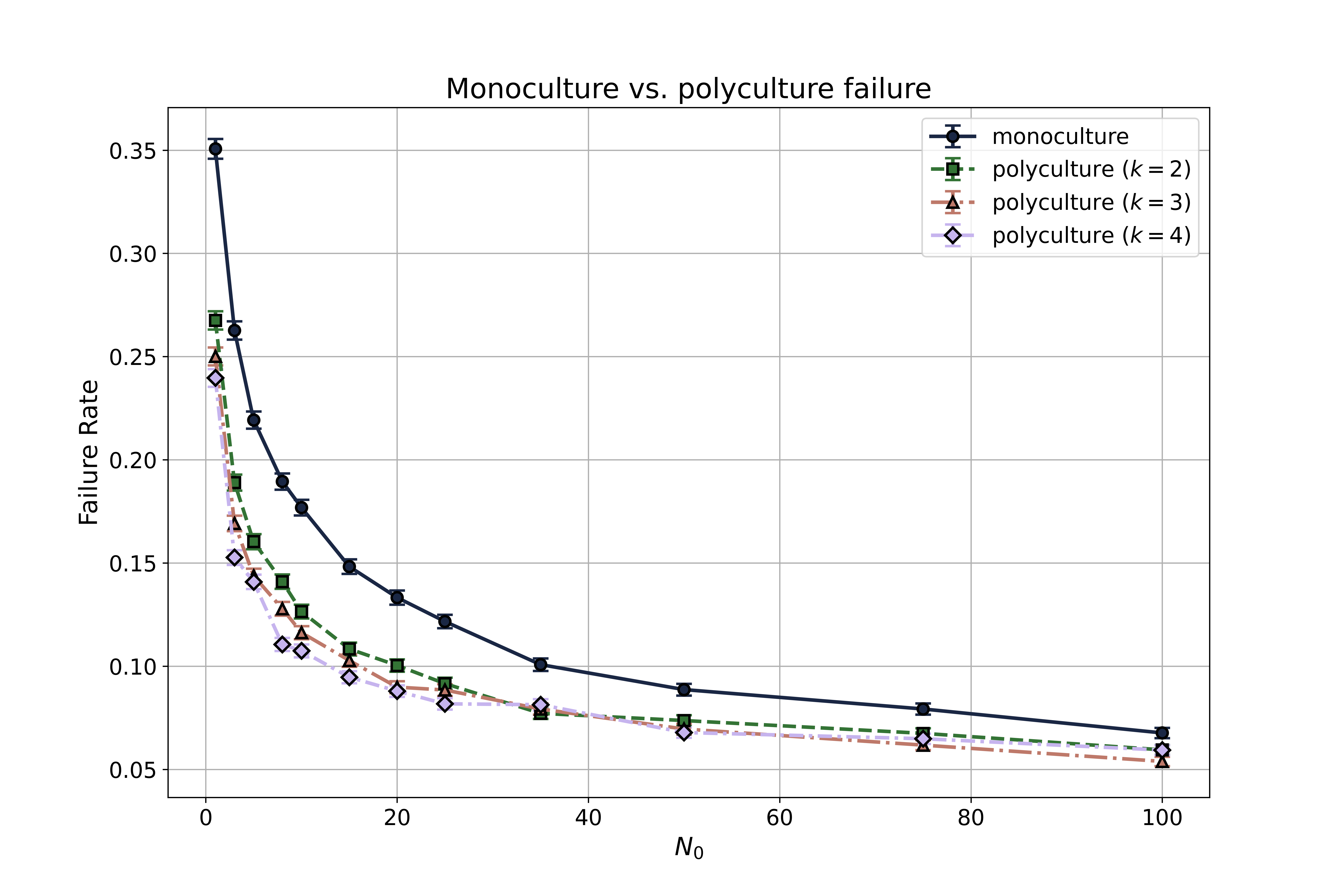}
    \caption{Failure to Identify Best Arm in Monoculture and Polyculture.}
    \label{fig:bandit-failure}
\end{figure}

\Cref{fig:bandit-failure} shows that monoculture consistently performs worse than polyculture across different values of $N_0$.
In simulations over a time horizon $T=1{,}000$ with mean rewards $\mu_1, \mu_2$ drawn independently from a Beta$(2, 2)$ prior, having more independent polyculture algorithms (larger $k$) leads to lower rates of failure to identify the best arm.
At least in this example, our intuition is borne out: polyculture mitigates the failure to explore.\footnote{See \citet{zollman2010epistemic} for related work on epistemic network structure and learning in bandit problems.}

\Cref{thm:main} states that this is true more generally.
For sufficiently long time horizons, polyculture identifies the best arm more often than monoculture, regardless of how $\mu_1, \mu_2$ are chosen.

\begin{theorem}[Informal]
  \label{thm:main-informal}
  For sufficiently large $T$, monoculture is more likely to fail to identify the best arm than polyculture.
\end{theorem}
\noindent We provide a formal statement and proof in Appendix~\ref{app:bandits}.

Should we expect to see exploration failures under monoculture in practice?
The argument that we should is intuitive: Monoculture may fail to identify some good candidate, but each polyculture algorithm has an independent chance of identifying them.
More independent chances at success yields a higher overall likelihood of success.
Each polyculture algorithm does its own implicit exploration, which leads to a better explore-exploit balance.
Moreover, qualitatively similar results hold in different models like sequential social learning~\citep{banerjee1992simple,bikhchandani1992theory}.

And yet, there are several practical considerations that blunt the force of this
objection. First, modern platforms do not engage in pure exploitation. Google
does not simply return the search results of highest estimated quality. Instead,
it experiments (or ``explores'') by showing results of lower estimated quality
precisely because exploration prevents information failures. Netflix, LinkedIn,
and most major online platforms do the same. In principle, users might prefer a platform, algorithm, or agent that maximizes their
individual expected utility instead of subjecting them to experiments. In
practice, this market force does not appear to be strong enough to force pure exploitation.

Moreover, even in a world with no explicit exploration, the bandits literature details a variety of conditions under which algorithms still acquire sufficient information.
Even small amounts of randomness or heterogeneity in the population suffice to yield exploration ``for free''~\citep{bastani2021mostly,kannan2018smoothed,raghavan2023greedy}.
Thus, if people have sufficiently heterogeneous preferences, or if a fraction of the population doesn't strictly follow the algorithm's recommendations, monoculture can still succeed.
Similarly, as long as a small fraction of the population is risk-seeking, they provide enough exploration to serve the entire population~\citep{banihashem2023bandit}.\footnote{Risk-averse users, by contrast, provide no benefit with respect to exploration \citep{banihashem2023bandit}.}

Finally, in the setting of \Cref{thm:main-informal}, any agent can pull any arm, no matter what the other agents do. This may accurately model some real-world settings where we might be concerned about algorithmic monoculture, and in these cases, the theorem gives a strong reason to favor polyculture. But it doesn't accurately model settings like hiring, where no candidate can be hired by multiple firms at the same time. Even if firms agree on the ranking of candidates, they cannot all hire their top-ranked candidate. Most, indeed all but one, have to go further down the list. In this way,  hiring involves externalities, whereby the decision of one firm affects the options available to others. And with a fixed unemployment rate, same number of candidates will be hired (and hence ``explored'') in any given round, regardless of whether we have monoculture or polyculture. Do the externalities inherent in the hiring setting force upon us sufficient exploration regardless of whether we have monoculture or polyculture?

To shed light on this question, we consider the following model.\footnote{See also \citet{badanidiyuru2018bandits} for a related type of bandit problem\textemdash bandits with knapsacks\textemdash in which the pulling of an arm yields externalities which must be taken into account. Theirs is a single-agent model where the agent has limited resources (time, money, etc), which get depleted as they pull more arms.} There are $n$ agents (the firms), $k>n$ arms (the candidates), and $t$ rounds. When an arm is pulled by an agent in a given round, this represents the corresponding candidate being hired by the corresponding firm for that round. The agents move in some order, each pulling their top-ranked arm from among those which have not yet been pulled in that round. All agents are greedy, pulling the arm with highest expected reward, given their information at the time. At the end of each round, all agents learn which arms were pulled by which agent, as well as what rewards they yielded, and they update in the usual Bayesian way.\footnote{One might object that it's implausible that all agents would learn about the rewards yielded by arms that \textit{other} agents pulled. Firms typically wouldn't publicly disclose the performance evaluations of their employees. We make this idealizing assumption for the sake of analytical and computational tractability. Without it, the actions of any given agent would provide the others with evidence about their belief state, which is in turn based on not only the rewards they directly observed, but also on inferences about other arms' rewards based on the actions of various other agents, and so on.}

In monoculture, agents have a shared prior for each arm's reward distribution, while in polyculture, they have independent priors.\footnote{We have already noted (\S\ref{subsec:aggregation}) that this independence assumption is implausible, but we make it here to ensure that we are considering a maximally attractive form of polyculture.} These priors are generated by giving agents some initial samples from each arm. In monoculture, agents are all given the same initial samples, while in polyculture, they are given distinct sets of initial samples.\footnote{Specifically, each arm $i$ gives reward 1 with probability $\mu_i$ and reward 0 otherwise, with each $\mu_i$ drawn independently from a Beta(2,2) distribution. An agent's prior for each arm's $\mu_i$ is generated by updating Beta(2,2) with $N_0=5$ initial samples from that arm. In monoculture, each agent's prior is generated using the \textit{same} $N_0$ initial samples from each arm. In polyculture, each agent's prior is generated using an \textit{independent} set of $N_0$ initial samples.} Because hiring order matters in polyculture (but not in monoculture), we  consider two different polyculture setups: one in which the hiring order is fixed throughout, and another in which the hiring order is randomized at the start of each round. 

We focus on two summary statistics of interest. The first is Total Bayesian Regret: the difference between (i) the total expected\footnote{Expectations here are relative to the arms' true reward distributions, not agents' beliefs.} reward (summed across all agents and all rounds) if all agents knew the arms' true reward distributions\footnote{This is the number of rounds times the sum of the means of the reward distributions of the best $n$ arms.} and (ii) the total expected reward (summed across all agents and all rounds) of the arms that actually got pulled. This is a measure of how suboptimal the collective performance of the agents was. The second statistic, which parallels the one considered in the context of \Cref{thm:main-informal}, is the number of arms that would be misclassified as being among the top $n$ by an impartial observer with access to all of the information available to any of the agents at the end of the rounds.\footnote{Our impartial observer starts off with Beta(2,2) priors for each arm, and then updates on all of the various agents' $N_0$ initial samples, as well as all actions and rewards for all rounds.} This is a measure of the extent to which the system generates enough information over the course of all of the rounds for society to then know which arms ``should'' be pulled. 

Results for simulations with $k=100$ arms, $T=200$ rounds, and various numbers
of agents or firms, averaged over 1,000 runs of the simulation, are shown in
\Cref{fig:regret,fig:misclassification}. Standard errors are too
small to show.

\begin{figure}[H]
    \centering
    \includegraphics[width=0.8\linewidth]{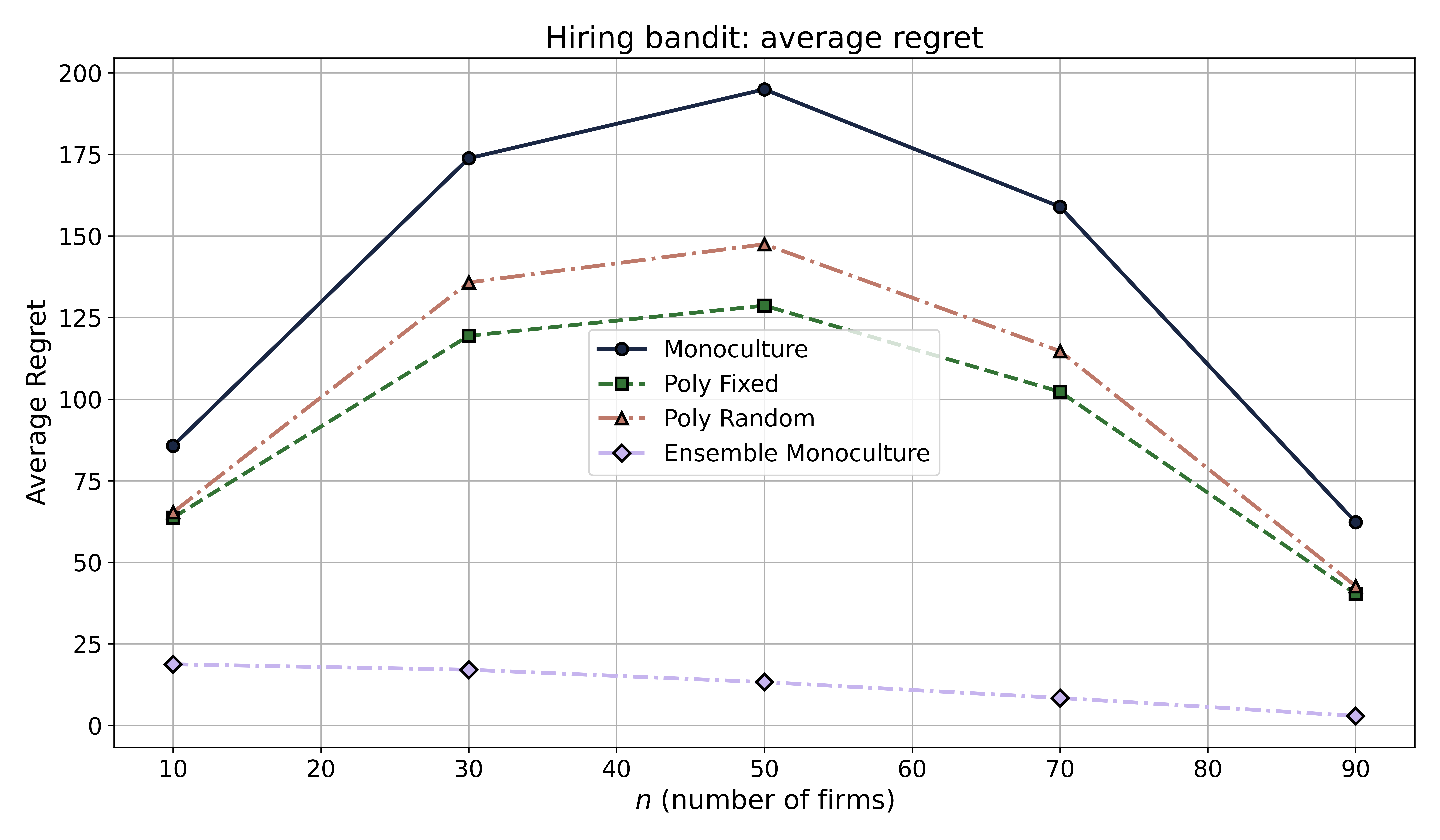}
    \caption{Total Bayesian Regret.}
    \label{fig:regret}
\end{figure}

\begin{figure}[H]
    \centering
    \includegraphics[width=0.8\linewidth]{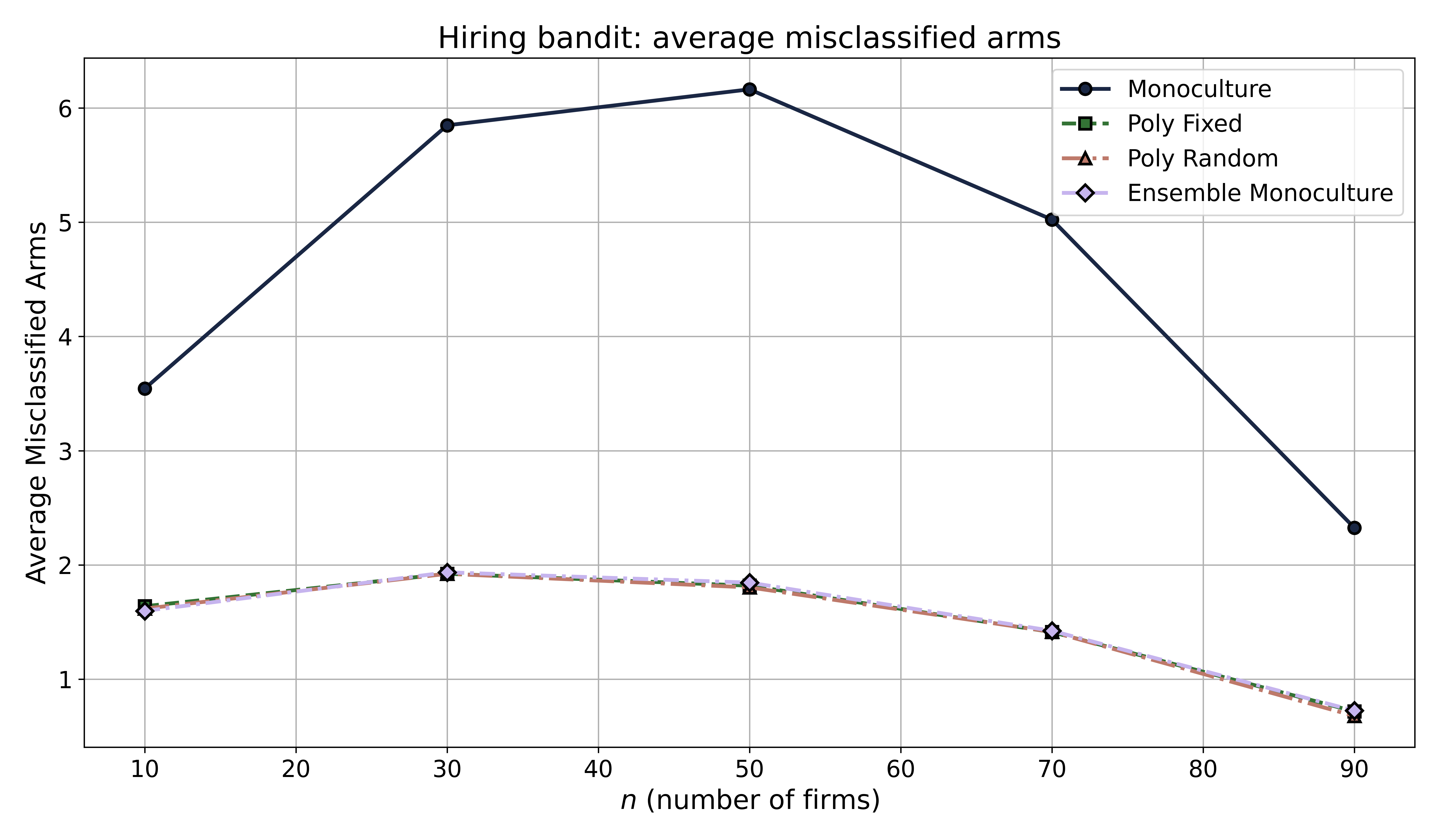}
    \caption{Number of Arms Misclassified by Impartial Observer.}
    \label{fig:misclassification}
\end{figure}

Polyculture\textemdash whether fixed-order or random-order\textemdash outperformed monoculture with respect to both Total Bayesian Regret and with respect to the number of arms misclassified by an impartial observer. This was true for all the different numbers of agents we considered. 

We take these simulations to show that even though the externalities associated with hiring force a certain amount of exploration, polyculture still performs better than monoculture by striking a more optimal balance of exploration and exploitation. At least, this is so when the algorithm deployed under monoculture is on a par with each of the various algorithms deployed under polyculture.

However, as we have seen, we can convert a given polyculture into a monoculture
through ensembling. So we also consider an ensemble monoculture setting in which
agents have a shared prior, derived from polyculture by pooling together all the
information initially available to \textit{any} of the agents in that
polyculture.\footnote{Specifically, in ensemble monoculture, each agent's prior
for a given arm's reward distribution is generated by updating Beta(2,2) on
\textit{all} of the initial $N_0$ samples seen by \textit{any} agent in the
polyculture of that run of the simulation. This models a setting where a
centralized platform observes all of the information previously used to form
each polyculture run's prior beliefs.} So, in ensemble monoculture and the
polyculture from which it was derived, the agents \textit{collectively} had the
same information. But that information was distributed in polyculture but pooled
in ensemble monoculture. Results for ensemble monoculture are also  shown in
\Cref{fig:regret,fig:misclassification}.

For every number of agents considered, ensemble monoculture achieved \textit{significantly} lower Total Bayesian Regret than not only ordinary monoculture, but also both forms of polyculture. And it performed essentially equally well as the two polyculture systems with respect to the number of arms that would be misclassified by an impartial observer. 

These results parallel the point made in \S\ref{subsec:aggregation}, namely that monoculture performs worse than polyculture if the former's lone algorithm is on a par with each of the latter's many algorithms. But if we turn a given polyculture into a monoculture through some sort of ensembling, the resulting monoculture can do even better than the original polyculture.

\section{Conclusion}
\label{sec:discussion}
Critics of algorithmic monoculture regard it as harmful and inefficient, or even dystopian. To evaluate their concerns, we considered several objections to monoculture: that it threatens to systemically exclude certain people from valuable opportunities such as employment, that it creates poor incentives and thereby diminishes or subverts people's agency, and that it is suboptimal from an informational perspective, failing to take advantage of `wisdom of crowds'-style mechanisms and striking a bad balance of exploration and exploitation. We find the first objection unconvincing. And while the other objections may cut against some forms of monoculture, other forms escape largely unscathed. Monoculture may not be such a bad thing after all. 

We highlight two themes of our discussion. First, monoculture may indeed be worse than polyculture when the sole algorithm deployed under monoculture is of a piece with the various algorithms deployed under polyculture, i.e., roughly as accurate, sophisticated, and so on.
In principle, some of monoculture's deficiencies might be remedied by an
\textit{ensembling} approach, in which the many algorithms which might be deployed independently in polyculture are instead packaged together into a single one through averaging or some other amalgamation mechanism. Indeed, an ensemble-based monoculture can outperform a corresponding polyculture, as our simulations illustrate. In practice, however, such ensembling may be infeasible---for example, in order to create an ensemble algorithm based on the polyculture emergent from idiosyncratic hiring managers across firms, we would need to solicit the opinions of each hiring manager, which might defeat the purpose of building an algorithm in the first place. Still, there may be ways for a monoculture to use randomization or personalization to bring back some of the benefits of polyculture.\footnote{For discussion of randomization, see \citet{jain2024scarce}). Note also that the idealized form of polyculture we've considered\textemdash one where the algorithms are uncorrelated, conditional on candidates' objective values\textemdash is unlikely to be realized in practice. As a polyculture's algorithms become more correlated, its informational advantages over ordinary monoculture will diminish.} 

Second, the externalities inherent in the hiring setting enable monoculture to avoid certain potential harms. Firms have a roughly fixed number of positions that they want to fill, and candidates can only accept one job offer. This prevents monoculture from being more exclusionary than polyculture, as roughly the same number of candidates will be left jobless under either system. It also forces monoculture to engage in a substantial amount of exploration; firms can’t
all hire the top-ranked candidate, so many firms will have to hire and thereby ‘explore’ lower-ranked candidates.

Similar externalities are present in the cases of lending and college admissions. To a first approximation, creditors have a certain amount of money that they want to lend, and borrowers can only take on a certain amount of debt. (E.g., once you take out one mortgage on your house, you typically can't take out another one on that same house.) Similarly, colleges have a certain number of places, and students can only attend one college at a time. So we should also not expect algorithmic monoculture to exclude more people from credit or college than polyculture. And we should expect monoculture to still explore borrowers or college applicants regarded as more uncertain. 

What about other settings which do not involve such externalities? As a toy example, consider film awards. There is nothing prohibiting all awards shows---the Oscars, Cannes, the BAFTAs, etc.---from giving the best actor award to the same person. So strict monoculture, where all award committees use the same ranking and hence give their award to the same actor, would indeed likely be more exclusionary than polyculture.\footnote{Of course, whether this is objectionable here depends on the details. If there is a fact of the matter about which actor was best, then perhaps that actor should receive all the awards. But it might be better to distribute awards more evenly, insofar as awards have diminishing marginal utility for recipients. }

But even though nothing \textit{prohibits} awards shows from all giving the best actor award to the same person, competitive pressures would likely mitigate the exclusionary effects of monoculture. If everyone knew that awards shows were all the same, people would only watch the first one. So to continue to be relevant, shows later in the cycle will want to at least consider doing something different. Even if they all used the same algorithm to come up with an initial ranking of actors, they would be incentivized to use this common, shared signal in different ways, which introduces subtle game-theoretic considerations. 

Something like this may happen with AI used to guide drug discovery. Even if all pharmaceutical firms used the same algorithm to decide which sorts of compounds are the most promising, they wouldn't want to herd around the same top-ranked compound, since only one firm can patent and market any given compound. So we should expect firms to use the shared signal provided by that algorithm in different ways. Something similar can arise in the use of generative AI, where users must balance their desire to produce the `best' content with the need to differentiate themselves from their competitors. See \citet{raghavan2024competition} for game-theoretic analysis of monoculture in generative AI. 

Much of our analysis has focused on domains that resemble hiring, where there are many decision-makers (the firms) with a fixed number of opportunities (jobs) to allocate, and where these opportunities are exclusive (since no one can hold multiple jobs at once). We leave it to future work to study the impacts of monoculture in other domains like medicine and criminal justice, which have importantly different structures.\footnote{See \citet{angwin2016machine} for discussion of criminal recidivism prediction algorithms.} Our discussion has also been theoretical, focusing on models that idealize away from various real-world complexities. For instance, firms don't all advertise their jobs at the same time, candidates don't apply to all firms, and hiring markets have friction and congestion, which can affect outcomes in complex ways \citep[see e.g.,][]{baek2025hiring}. 

We encourage further theoretical work that models these and other real-world complexities, as well as empirical work, ideally involving `natural experiments' with shifts toward or away from monoculture.  Still, we have shown that monoculture is not the inevitable disaster its critics fear; it is a type of structure which allows for refinement and improvement, e.g., through ensembling. Realizing its potential benefits, however, requires a commitment to intentional design and continued study of its consequences, both good and bad.

\newpage

\bibliographystyle{plainnat}
\bibliography{monoculture.bib}
\newpage 

\appendix
\section{Polyculture vs. Monoculture in Multi-Arm Bandits}
\label{app:bandits}

\paragraph*{Notation.}

We'll consider a two-arm bandits setup, where the arms have mean rewards $\mu_1$ and $\mu_2$ respectively.
Assume that $1 > \mu_1 > \mu_2 > 0$, and let $\Delta = \mu_1 - \mu_2 > 0$.
Rewards are binary, so when arm $j$ is pulled, the realized reward is 1 with probability $\mu_j$ and 0 with probability $1-\mu_j$.

Rounds proceed sequentially over timesteps $t = 1, 2, \dots, T$.
As is standard, we will assume independence between random events at time $t$ and all randomness at times $\tau < t$.\footnote{We could make this more formal by defining a filtration $\mc F_t$ over the relevant $\sigma$-algebra. For clarity, we omit these standard definitions.}
Under these assumptions, we can specify the full distribution over rewards as follows.
Consider two infinite \textit{reward schedules} $\{X_{1,n}\}_{n=1}^{\infty}$ and $\{X_{2,n}\}_{n=1}^{\infty}$, where $X_{j, n}$ is the reward for arm $j$ on the $n$th time it is selected.
For each arm $j$, let $Z_j(n) = \sum_{\ell=1}^n X_{j,\ell}$ denote
the total reward from its first $n$ pulls. For the $i$th polyculture run, the
analogous quantity is $Z_j^{(i)}(n)$.
Because of our independence assumption, the $X_{j, n}$'s are independent and identically distributed as Bernoulli random variables with parameter $\mu_j$.
We will refer to the combination of these two reward schedules as $\mc X$, and $\PX$ denotes its probability law (i.e.,~ independent Bernoullis).

The final source of randomness in our setup comes from an initial set of observations that all agents have access to.
Equivalently, we could model this as a common prior; instead, we choose to study a fully frequentist setting, where an initial shared history $\hist$ effectively takes the role of the prior.
Following \citet{banihashem2023bandit}, we assume $\hist$ consists of $N_0 \ge 1$ samples drawn iid from each arm.
Let $S_j(0)$ be the sum of those samples, i.e.,~the number of ``heads'' observed for arm $j$ in that initial set of samples $\hist$.
Intuitively, larger $N_0$ values correspond to stronger and more informative priors.
We'll use $\PH$ to refer to the probability law over $\hist$.
As we'll show, behavior in our model will be deterministic given $\hist, \mc X$.

We study a setting with $k \ge 2$ agents. Each agent will run a frequentist greedy
bandit algorithm, meaning they simply select the arm with the highest observed
empirical mean so far, breaking ties in favor of arm 2 by convention.
More formally, let $n_j(t) = n_j(t; \hist, \mc X)$ be the number of times arm $j$ has been selected up to and including time $t$.
Let $S_j(t) = S_j(0) + Z_j(n_j(t))$.
The empirical mean for arm $j$ at time $t$ is $\hat \mu_j(t) = \frac{S_j(t)}{N_0 + n_j(t)}$.
We'll let $A_t = A_t(\hist, \mc X)$ denote the arm chosen at time $t$.
Let $\PH$ be the distribution of $\hist$, let $\PX$ be the distribution of
$\mc X$, and let $\Pboth$ be their joint distribution.

Our goal is to model both monoculture and polyculture in this setting.
We do so via the information structure.
Under monoculture, the agent at time $t$ can observe all prior actions and rewards.
Under polyculture, this agent can only observe their own actions and rewards but not those of other agents.
Thus, under our model, monoculture is equivalent to a single bandit run, while polyculture is simply $k$ independent bandit runs (with shared initial samples $\hist$).
We will use $\vX$ to denote the vector of $k$ polyculture runs.

We are interested in the probability that after the last time $T$, the observed empirical mean reward of the objectively worse arm is at least as high as the observed empirical mean reward of the better. This would mean that someone who knew the whole history would fail to identify the better arm as uniquely best. $\mc F_{\mono}(T)$ is the event that under monoculture, the observed empirical mean reward of arm 2 is at least as high as that of arm 1 (i.e., $\hat \mu_2(T) \ge \hat \mu_1(T)$). Under polyculture, the all-seeing observer pools all $k$ runs' data; $\mc F_{\poly}(T)$ is the event that this observer fails to identify the better arm as uniquely best, i.e.,\ $\mutall(T) \ge \muoall(T)$, where
\begin{align*}
  \mujall(T) = \frac{S_j(0) + \sum_{i \in [k]} Z_j^{(i)}(n_j^{(i)}(T))}{N_0 + \sum_{i \in [k]} n_j^{(i)}(T)}
\end{align*}
is the pooled empirical mean for arm $j$ across all $k$ runs.
Here $T$ indexes the horizon of each bandit run; because our comparison is
asymptotic, the result is unchanged if $T$ instead denotes a global horizon
allocated round-robin across the $k$ runs.

\paragraph*{Key events.}

We define several events used throughout the proofs below.

For a single bandit run $(\mc H_0, \mc X)$, define the \emph{lock-in time}
\begin{align*}
  L = L(\mc H_0, \mc X) = \inf\bigl\{t \ge 1 : A_{t'} = A_t \text{ for all } t' \ge t\bigr\},
\end{align*}
with $L = \infty$ if no such $t$ exists. Intuitively, $L$ is the first time from
which the greedy never switches arms again; $A_L$ is the arm locked into. We
write $R(t; \mc H_0, \mc X) = \{L < t/2\} \cap \{A_L = 1\}$ for the event that the run has locked in to the better arm by time $t/2$ (we drop $t$ and write $R$ when referring to the asymptotic event $\{L < \infty\} \cap \{A_L = 1\}$). The wrong-lock-in event is
\begin{align*}
  W = W(\mc H_0, \mc X) = \{L < \infty\} \cap \{A_L = 2\},
\end{align*}
and we write $W(t; \mc H_0, \mc X) = \{L < t/2\} \cap \{A_L = 2\}$ for its truncated version.

For polyculture with $k$ independent runs, let $L^{(i)} = L(\mc H_0, \mc X_i)$ and $A_{L^{(i)}}^{(i)}$ denote the lock-in time and locked-in arm of the $i$-th run. We define the following partition of the event $\{\max_{i \in [k]} L^{(i)} < t/2\}$ (all runs have locked in by time $t/2$):
\begin{align*}
  \allwrong(t) &= \bigcap_{i \in [k]} \bigl\{L^{(i)} < t/2\bigr\} \cap \bigl\{A_{L^{(i)}}^{(i)} = 2\bigr\}, \\
  \allright(t) &= \bigcap_{i \in [k]} \bigl\{L^{(i)} < t/2\bigr\} \cap \bigl\{A_{L^{(i)}}^{(i)} = 1\bigr\}, \\
  \somewrong(t) &= \Bigl(\bigcap_{i \in [k]} \bigl\{L^{(i)} < t/2\bigr\}\Bigr) \setminus \bigl(\allwrong(t) \cup \allright(t)\bigr).
\end{align*}
That is, $\allwrong(t)$ is the event that every run locks in to the wrong arm, $\allright(t)$ that every run locks in to the correct arm, and $\somewrong(t)$ that all runs have locked in but at least one locked in to each arm.
We write $\allright = \bigcup_{t=1}^{\infty}\allright(t)$ and
$\somewrong = \bigcup_{t=1}^{\infty}\somewrong(t)$ for the corresponding
asymptotic events.

\paragraph*{Statement and proof of \Cref{thm:main}.}

With this notation, we formally state and prove the result.

\addtocounter{theorem}{-1}
\begin{theorem}[Formal]
  \label{thm:main}
  For sufficiently large $T$,
  \begin{align*}
    \Pbothk(\mc F_{\mono}(T)) > \Pbothk(\mc F_{\poly}(T)).
  \end{align*}
\end{theorem}

\begin{proof}
  For any fixed $\mc H_0$, by construction,
  \begin{align*}
    \allwrong(t) \cup \somewrong(t) \cup \allright(t) = \Bigl\{\max_{i \in [k]} L^{(i)} < t/2\Bigr\}.
  \end{align*}
  Moreover, these events are mutually disjoint. Thus,
  \begin{align*}
    \PXk(\mc F_{\poly}(t; \mc H_0))
    &= \PXk(\mc F_{\poly}(t; \mc H_0) \cap \allwrong(t; \mc H_0)) \\
    &+ \PXk(\mc F_{\poly}(t; \mc H_0) \cap \somewrong(t; \mc H_0)) \\
    &+ \PXk(\mc F_{\poly}(t; \mc H_0) \cap \allright(t; \mc H_0)) \\
    &+ \PXk\p{\mc F_{\poly}(t; \mc H_0) \cap \left\{\bigcup_{i \in
    [k]} L^{(i)} \ge t/2\right\}}
  \end{align*}
  We bound each of these terms separately in the limit as $t \to \infty$. By
  \Cref{lem:all-wrong},
  \begin{align*}
    \limsup_{t \to \infty} \PXk(\mc F_{\poly}(t; \mc H_0) \cap
    \allwrong(t; \mc H_0))
    &\le \lim_{t \to \infty} \PXk(\allwrong(t; \mc H_0)) \\
    &= \PX(W; \mc H_0)^k.
  \end{align*}
  By \Cref{lem:some-wrong},
  \begin{align*}
    \lim_{t \to \infty} \PXk(\somewrong(t; \mc H_0) \cap \mc
    F_{\poly}(t; \mc H_0)) = 0.
  \end{align*}
  By \Cref{lem:all-right},
  \begin{align*}
    \lim_{t \to \infty} \PXk(\allright(t; \mc H_0) \cap \mc
    F_{\poly}(t; \mc H_0)) = 0.
  \end{align*}
  Finally, we use the fact that $\lim_{t \to \infty} \PX(L \ge t/2) = 0$
  (\Cref{lem:lock-in}) to get
  \begin{align*}
    \lim_{t \to \infty}
    \PXk\p{\mc F_{\poly}(t; \mc H_0) \cap \left\{\bigcup_{i \in
    [k]} L^{(i)} \ge t/2\right\}}
    &\le 
    \lim_{t \to \infty}
    \PXk\p{\bigcup_{i \in [k]} L^{(i)} \ge t/2} \\
    &= 0.
  \end{align*}

  Putting these all together, for fixed $\mc H_0$,
  \begin{align*}
    \limsup_{t \to \infty} \PXk(\mc F_{\poly}(t; \mc H_0))
    &\le \PX(W; \mc H_0)^k.
  \end{align*}
  Because all rewards are iid, we can ``stack'' the rewards from $\vX$ into a
  single run $\mc X$ when considering monoculture. Thus,
  \begin{align*}
    \Pbothk(\mc F_{\mono}(T)) = \Pboth(\mc F_{\mono}(T)).
  \end{align*}

  By \Cref{lem:mono-fail}, for any fixed $\mc H_0$,
  \begin{align*}
    \lim_{t \to \infty} \PX(\mc F_{\mono}(t; \mc H_0))
    = \PX(W; \mc H_0).
  \end{align*}
  Combining these,
  \begin{align*}
    \lim_{t \to \infty} \PX(\mc F_{\mono}(t; \mc H_0))
    &= \PX(W; \mc H_0) \\
    &\ge \PX(W; \mc H_0)^k \\
    &\ge \limsup_{t \to \infty} \PXk(\mc F_{\poly}(t; \mc H_0)).
  \end{align*}
  Since all quantities lie in $[0, 1]$, the Dominated Convergence Theorem
  applies to the monoculture term, while the reverse Fatou lemma applies to the
  polyculture term:
  \begin{align*}
    \lim_{t \to \infty} \Pbothk(\mc F_{\mono}(t))
    &= \EE{\mc H_0}{\lim_{t \to \infty} \PX(\mc F_{\mono}(t; \mc H_0))}
    = \EE{\mc H_0}{\PX(W; \mc H_0)}, \\
    \limsup_{t \to \infty} \Pbothk(\mc F_{\poly}(t))
    &\le \EE{\mc H_0}{\limsup_{t \to \infty} \PXk(\mc F_{\poly}(t; \mc H_0))}
    \le \EE{\mc H_0}{\PX(W; \mc H_0)^k}.
  \end{align*}
  Since $p \ge p^k$ for all $p \in [0,1]$ and $k \ge 2$, we have
  $\EE{\mc H_0}{\PX(W; \mc H_0)} \ge \EE{\mc H_0}{\PX(W; \mc H_0)^k}$,
  giving the weak inequality.

  For strict inequality: by \Cref{lem:inf-prob},
  $\PH\bigl(\PX(W; \mc H_0) \notin \{0, 1\}\bigr) > 0$.
  For any $\mc H_0$ with $\PX(W; \mc H_0) \in (0,1)$, we have
  $\PX(W; \mc H_0) > \PX(W; \mc H_0)^k$.
  Since this strict inequality holds on a positive-measure set of $\mc H_0$'s,
  \begin{align*}
    \EE{\mc H_0}{\PX(W; \mc H_0)} > \EE{\mc H_0}{\PX(W; \mc H_0)^k},
  \end{align*}
  and hence
  \begin{align*}
    \lim_{t \to \infty} \Pbothk(\mc F_{\mono}(t))
    > \limsup_{t \to \infty} \Pbothk(\mc F_{\poly}(t)),
  \end{align*}
  completing the proof.
\end{proof}

\begin{lemma}
  \label{lem:all-wrong}
For any
  $\mc H_0$,
  \begin{align*}
    \lim_{t \to \infty} \PXk(\allwrong(t; \mc H_0)) = \PX(W; \mc
    H_0)^k.
  \end{align*}
\end{lemma}
\begin{proof}
  Because the $k$ bandit runs are independent conditioned on $\mc H_0$,
  \begin{align*}
    \PXk(\allwrong(t; \mc H_0))
    &= \prod_{i=1}^k \PX(W(t); \mc H_0) \\
    &= \PX(W(t); \mc H_0)^k
  \end{align*}
  Taking the limit as $t \to \infty$ and applying \Cref{lem:wrong-converge}
  completes the proof.
\end{proof}

\begin{lemma}
  \label{lem:some-wrong}
For any $\mc H_0$,
  \begin{align*}
    \lim_{t \to \infty} \PXk(\somewrong(t; \mc H_0) \cap \mc
    F_{\poly}(t; \mc H_0)) = 0.
  \end{align*}
\end{lemma}
\begin{proof}
  Assume $\PXk(\somewrong;\mc H_0)>0$ (otherwise the result is immediate).
  We'll use the fact that under $\somewrong$, the all-seeing observer sees
  infinitely many samples (in the limit) of both arms. This is because at least
  one run has locked in to each of the two arms, meaning it will pull that arm
  forever, generating infinitely many samples. This means that the observer's
  estimates of the arm rewards converge in probability to their true values.
  Formally, for each arm $j$, on $\somewrong$,
  \begin{align*}
    \lim_{t \to \infty} \sum_{i \in [k]} n_j^{(i)}(t; \mc H_0, \vec{\mc X})
    = \infty.
  \end{align*}
Here, we use the fact that
  \begin{align*}
    \mujall(t)
    &= \frac{S_j(0) + \sum_{i \in [k]} Z_j^{(i)}(n_j^{(i)}(t))}{N_0 + \sum_{i \in [k]}
    n_j^{(i)}(t)} .
  \end{align*}
  By the strong law of large numbers (which holds conditioned on any
  positive-measure event),
  \begin{align*}
    \mujall(t) \given \somewrong \xrightarrow{a.s.} \mu_j,
  \end{align*}
  since it is the empirical mean of infinitely many iid samples (the initial
  $S_j(0)$ samples get ``washed out'' in the limit). But this means that the
  all-seeing observer gets infinitely precise estimates of the true means with
  probability 1, meaning they must identify the correct arm with probability 1
  in the limit. Formally,
  \begin{align*}
\lim_{t \to \infty}
     &\PXk(\mc F_{\poly}(t; \mc H_0) \given
    \somewrong) \\
     &= \lim_{t \to \infty} \PXk(\muoall(t) \le \mutall(t) \given
    \somewrong) \\
     &\le \lim_{t \to \infty} \PXk(|\muoall(t) - \mu_1| \ge \Delta/2 \cup
     |\mutall(t) -
     \mu_2| \ge \Delta/2 \given \somewrong)
     \tag{$\mc F_{\poly}$ implies that at least one of the empirical means is
     wrong by $\ge \Delta/2$}
     \\
     &= 0,
  \end{align*}
  Since $\somewrong(t)\subseteq\somewrong$,
  \begin{align*}
    \PXk(\somewrong(t; \mc H_0) \cap \mc F_{\poly}(t; \mc H_0))
    \le \PXk(\somewrong \cap \mc F_{\poly}(t; \mc H_0)) \longrightarrow 0
  \end{align*}
  as desired.
\end{proof}

\begin{lemma}
  \label{lem:all-right}
For any $\mc H_0$,
  \begin{align*}
    \lim_{t \to \infty} \PXk(\allright(t; \mc H_0) \cap \mc
    F_{\poly}(t; \mc H_0)) = 0.
  \end{align*}
\end{lemma}
\begin{proof}
  In this case, all polyculture runs converge on arm 1 in the limit. This means
  that $\hat \mu_1^{(i)}(t) > \hat \mu_2^{(i)}(t)$ for all $i$, $t > L$.
  Naively, we might hope that this implies $\muoall(t) > \mutall(t)$ for $t >
  L$, completing the proof. This fails for two reasons.
  \begin{enumerate}
    \item In general, $a_i > b_i$ for all $i$ does not imply that same
      inequality applies to their weighted averages (this is Simpson's Paradox).
    \item $\mujall$ is not simply a weighted average of $\hat \mu_j^{(i)}$. A
      naive weighted average would ``count'' the initial history $\mc H_0$ $k$
      times, while $\mujall$ only includes a single copy.
  \end{enumerate}

  If $\PXk(\allright;\mc H_0)=0$, the result is immediate. Hence assume
  $\PXk(\allright;\mc H_0)>0$.
  We'll develop a more sophisticated strategy to deal with these issues. We
  begin by fixing $\mc H_0$. By \Cref{lem:lock-in-upper-bound}, when a bandit
  run locks in to arm 1 (i.e., on the event $R$),
  \begin{align*}
    \lim_{t \to \infty}
    \PX(\hat \mu_2(t; \mc H_0) < \mu_1 \given
    R) = 1.
  \end{align*}
  Therefore, with probability 1 in the limit, each polyculture agent $i$'s
  estimate for arm 2 is strictly smaller than $\mu_1$, i.e.,
  \begin{align*}
  \hat \mu_2^{(i)}(t; \mc H_0, \mc X_i) < \mu_1.
  \tag{$\forall i \in [k]$}
  \end{align*}

  Next, we show that this implies $\mutall(t) < \mu_1$ for sufficiently large
  $t$. Let
  \begin{align*}
    a
    &= k S_2(0) \\
    b
    &= k N_0 \\
    c
    &= \sum_{i \in [k]} Z_2^{(i)}(n_2^{(i)}(t)) \\
    d
    &= \sum_{i \in [k]} n_2^{(i)}(t).
  \end{align*}
  Then,
  \begin{align*}
    \mutall(t)
    &= \frac{\frac{1}{k} a + c}{\frac{1}{k} b + d} .
\end{align*}
  The naive (weighted) mean is
  \begin{align*}
    \hat \mu_2^{\text{naive}}(t)
    &= \frac{a+c}{b+d}.
  \end{align*}
  Note that if each $\hat\mu_2^{(i)}(t) < \mu_1$, then $\hat
  \mu_2^{\text{naive}}(t) < \mu_1$: each inequality $\hat\mu_2^{(i)}(t) < \mu_1$
  is equivalent to $S_2(0) + Z_2^{(i)}(n_2^{(i)}(t)) < \mu_1(N_0 + n_2^{(i)}(t))$, and
  summing over all $i \in [k]$ gives $a + c < \mu_1(b + d)$, i.e.,\
  $\hat\mu_2^{\text{naive}}(t) < \mu_1$.
  On $\allright$, for all sufficiently large $t$, arm 2 is the
  lower-mean arm in every run. Applying \Cref{lem:greedy-min} to each run and
  summing therefore gives $a/b\ge c/d$ for all sufficiently large $t$.
  Consequently, by \Cref{lem:abcd},
  \begin{align*}
    \PXk(\mutall(t)\le\hat\mu_2^{\text{naive}}(t)\given\allright)
    \longrightarrow 1.
  \end{align*}
  Therefore,
  \begin{align*}
    \PXk(\mutall(t; \mc H_0) < \mu_1 \given \allright)
    &\ge
    \PXk(\hat\mu_2^{\text{naive}}(t; \mc H_0) < \mu_1 \given \allright)
    -o(1) \\
    &\ge \PXk\p{\bigcap_{i \in [k]}\hat \mu_2^{(i)}(t; \mc H_0) < \mu_1
    \given \allright} -o(1) \\
    &= \prod_{i \in [k]} \PXk(\hat \mu_2(t; \mc H_0, \mc X_i) < \mu_1 \given
    R^{(i)}) -o(1) \\
    &= \PX(\hat \mu_2(t; \mc H_0) < \mu_1 \given
    R)^k -o(1)
  \end{align*}
  Taking $t \to \infty$ on both sides and applying
  \Cref{lem:lock-in-upper-bound} yields
  \begin{align*}
    \lim_{t \to \infty} \PXk(\mutall(t; \mc H_0) < \mu_1 \given
    \allright)
    &= 1.
  \end{align*}

  Finally, by the strong law of large numbers, as $t \to \infty$,
  \begin{align*}
    \muoall(t) \given \allright \overset{a.s.}{\longrightarrow} \mu_1.
  \end{align*}
  (Note that the SLLN holds even conditioned on some event as long as that event
  occurs with strictly positive probability.)
  On $\allright$, $\mutall(t)$ is eventually constant. Write
  $\mutall(\infty)$ for its eventual value. The preceding argument gives
  $\mutall(\infty)<\mu_1$ almost surely.
  Therefore, we can apply \Cref{lem:inf-inequality} with
  $Y_t=\mu_1-\muoall(t)$ and
  $\delta=\mu_1-\mutall(\infty)$ to get
  \begin{align*}
    \lim_{t \to \infty} \PXk(\muoall(t) > \mutall(t) \given \allright) = 1.
  \end{align*}
  As a result,
  \begin{align*}
    \lim_{t \to \infty} \PXk(\mc F_{\poly}(t; \mc H_0) \given
    \allright)
    &= \lim_{t \to \infty} \PXk(\mutall(t; \mc H_0) \ge \muoall(t; \mc H_0)
    \given \allright) \\
    &= 0.
  \end{align*}
  Since $\allright(t)\subseteq\allright$,
  \begin{align*}
    \PXk\bigl(\mc F_{\poly}(t;\mc H_0)\cap\allright(t)\bigr)
    \le \PXk\bigl(\mc F_{\poly}(t;\mc H_0)\cap\allright\bigr)
    \longrightarrow 0,
  \end{align*}
  as desired.
\end{proof}

\begin{lemma}
  \label{lem:mono-fail}
  For any $\mc H_0$,
  \begin{align*}
    \lim_{t \to \infty} \PX(\mc F_{\mono}(t; \mc H_0))
    = \PX(W; \mc H_0).
  \end{align*}
\end{lemma}
\begin{proof}
  \begin{align*}
    \PX(\mc F_{\mono}(t; \mc H_0))
    &= \PX(\hat \mu_1(t; \mc H_0) \le \hat
    \mu_2(t; \mc H_0)) \\
    &= \PX(\{\hat \mu_1(t; \mc H_0) \le \hat
    \mu_2(t; \mc H_0)\} \cap \{L < t/2\}) \\
    &+ \PX(\{\hat \mu_1(t; \mc H_0) \le \hat
    \mu_2(t; \mc H_0)\} \cap \{L \ge t/2\})
\end{align*}
  By \Cref{lem:lock-in} and the Monotone Convergence Theorem, $\lim_{t \to
  \infty} \PX(L \ge t/2) = \PX(L = \infty) = 0$. Therefore,
  \begin{align*}
    \lim_{t \to \infty}
    \PX(\mc F_{\mono}(t; \mc H_0))
    &= \lim_{t \to \infty}
    \PX(\{\hat \mu_1(t; \mc H_0) \le \hat
    \mu_2(t; \mc H_0)\} \cap \{L < t/2\}) \\
    &= \lim_{t \to \infty} \PX(\{A_L = 2\} \cap \{L < t/2\}) \\
    &= \lim_{t \to \infty} \PX(W(t)) \\
    &= \PX(W) \tag{by \Cref{lem:wrong-converge}}
  \end{align*}
\end{proof}

\begin{lemma}
  \label{lem:inf-prob}
  \begin{align*}
    \PH(\PX(W; \mc H_0) \notin \{0, 1\}) > 0.
  \end{align*}
\end{lemma}
\begin{proof}
  Because every $\mc H_0$ is sampled with positive probability, it suffices to
  show that $\exists \mc H_0$ such that $\PX(W; \mc H_0) \in (0, 1)$. In
  fact, we will show that this is true for every $\mc H_0$ such that $S_1(0),
  S_2(0) \ne 0$. For any such $\mc H_0$, we will show that the set of $\mc
  X$ that locks in to either arm 1 or arm 2 has nonzero measure.

  In particular, lock-in to arm $j$ occurs if:
  \begin{itemize}
    \item $\hat \mu_j(t) \ge \mu_j - c$ for all $t > t_0$ (the empirical mean
      for the lock-in arm is reasonably high).
    \item $\hat \mu_{j'}(t_0) < \mu_j - c$ (the empirical mean for the other arm
      is sufficiently low).
  \end{itemize}
  We will show that these happen simultaneously with positive probability.
  Fix $\delta\in(0,1)$ and choose $k_0$ sufficiently large that
  $\mu_j-\sqrt{2\log(1/\delta)/k_0}-N_0/(k_0+N_0)>0$.
  First, we apply \Cref{lem:blackwell-av-bound} to the empirical reward sequence
  for arm $j$ to get
  \begin{align*}
    \PX\p{\exists k > k_0 : \frac{Z_j(k)}{k} \le \mu_j -
    \sqrt{\frac{2 \log(1/\delta)}{k_0}}}
    &\le \delta
  \end{align*}
  Moreover, this probability only decreases if we condition on the event that
  the first $k_0$ rewards for arm $j$ are realized to be 1, since this weakly
  increases $Z_j(k)$ for all $k$. Thus,
  \begin{align*}
    \PX\p{\exists k > k_0 : \frac{Z_j(k)}{k} \le \mu_j - \sqrt{\frac{2
    \log(1/\delta)}{k_0}} \given X_{j,1} = X_{j,2} = \dots = X_{j,k_0} = 1}
    &\le \delta.
  \end{align*}
  We'll also consider the event that the first $k_1$ rewards from arm $j'$ are
  realized to be 0. For sufficiently large $k_1$, we can drive $Z_{j'}(k_1)/k_1$
  to be arbitrarily close to 0. In particular, we will choose $k_1$ such that
  \begin{equation}
    \label{eq:k1-def}
    \frac{S_{j'}(0)}{k_1 + N_0}
    < \min\p{\frac{S_j(0)}{N_0}, \mu_j - \sqrt{\frac{2
    \log(1/\delta)}{k_0}} - \frac{N_0}{k_0 + N_0} }.
  \end{equation}
  We have thus defined three events:
  \begin{align*}
    \mc E_1:
    &X_{j,1} = X_{j,2} = \dots = X_{j,k_0} = 1 \\
    \mc E_2:
    & \not \exists k > k_0 : \frac{Z_j(k)}{k} \le \mu_j - \sqrt{\frac{2
    \log(1/\delta)}{k_0}} \\
    \mc E_3:
    &X_{j',1} = X_{j',2} = \dots = X_{j',k_1} = 0
  \end{align*}
  The intersection of these events occurs with positive probability, and
  together, we will argue that they imply lock-in to arm $j$.

  To do so, observe that the first $k_1$ times that arm $j'$ is chosen, its
  empirical mean will monotonically decrease. On the $k_1$th time it is chosen,
  its empirical mean (including the $N_0$ original samples) will be
  $S_{j'}(0)/(k_1 + N_0)$. By~\eqref{eq:k1-def}, this will eventually drop below
  $S_j(0)/N_0$, so arm $j$ will be selected at some point. Because the first $k_0$
  realized rewards for arm $j$ are 1, once it is selected, it will be selected
  for at least the next $k_0$ times. Let $t_0$ be the last of these $k_0$
  initial selections for arm $j$. Under $\mc E_2$, for all $t > t_0$,
  \begin{align*}
    \hat \mu_j(t)
    &= \frac{n_j(t)}{n_j(t) + N_0}
    \frac{Z_j(n_j(t))}{n_j(t)} +
    \frac{N_0}{n_j(t) + N_0} \frac{S_j(0)}{N_0} \\
    &\ge \frac{n_j(t)}{n_j(t) + N_0}
    \p{\mu_j - \sqrt{\frac{2
    \log(1/\delta)}{k_0}}} + \frac{N_0}{n_j(t) + N_0} \p{\mu_j - \mu_j +
    \frac{S_j(0)}{N_0}} \\
    &\ge \mu_j - \sqrt{\frac{2 \log(1/\delta)}{k_0}} - \frac{N_0}{n_j(t) + N_0}
    \p{\mu_j - \frac{S_j(0)}{N_0}} \\
    &\ge \mu_j - \sqrt{\frac{2 \log(1/\delta)}{k_0}} - \frac{N_0}{k_0 + N_0}.
  \end{align*}
  Again using~\eqref{eq:k1-def}, once arm $j'$ has been chosen $k_1$ times, its
  empirical mean will be below $\hat \mu_j(t)$ for all $t > t_0$. Thus, we have
  established that under $\mc E_1, \mc E_2, \mc E_3$, arm $j'$ cannot be chosen
  more than $k_1$ times, meaning we have locked in to arm $j$ with positive
  probability. Because this holds for both $j \in \{1, 2\}$,
  \begin{align*}
    \PH(\PX(W; \mc H_0) \notin \{0, 1\}) > 0
  \end{align*}
  as desired, completing the proof.
\end{proof}

\begin{lemma}
  \label{lem:lock-in}
  \begin{align*}
    \PX(L < \infty) = 1.
  \end{align*}
\end{lemma}
\begin{proof}
  Under $\{L = \infty\}$, by definition, the greedy switches between arms
  infinitely often. This means that each arm is chosen infinitely often,
  implying that $\lim_{t \to \infty} n_j(t) = \infty$ for both arms $j \in \{1,
  2\}$. Observe that
  \begin{align*}
    \hat \mu_j(t)
    &= \frac{S_j(0) + Z_j(n_j(t))}{N_0 +
    n_j(t)} \\
&= \frac{n_j(t)}{N_0 + n_j(t)} \frac{Z_j(n_j(t))}{n_j(t)} +
    \frac{S_j(0)}{N_0 + n_j(t)}
  \end{align*}
  By the strong law of large numbers,
  \begin{align*}
    \PX\p{\lim_{n \to \infty}\frac{Z_j(n)}{n} = \mu_j} = 1.
  \end{align*}
  Since $n_j(t) \to \infty$ as $t \to \infty$,
  \begin{align*}
    \lim_{n \to \infty}\frac{Z_j(n)}{n}
    = \lim_{t \to \infty}\frac{Z_j(n_j(t))}{n_j(t)}.
  \end{align*}
  Therefore, with probability 1,
  \begin{align*}
    \lim_{t \to \infty} \hat \mu_j(t)
    &=\lim_{t \to \infty} \frac{n_j(t)}{N_0 + n_j(t)} \frac{Z_j(n_j(t))}{n_j(t)} + \frac{S_j(0)}{N_0 + n_j(t)} \\
    &= \mu_j.
  \end{align*}
  This means that with probability 1, for sufficiently large $t$, $|\hat
  \mu_j(t) - \mu_j| < \Delta/2$. This means that for all sufficiently large $t$,
  \begin{align*}
    \hat \mu_1(t)
    &> \mu_1 - \frac{\Delta}{2} \\
    &= \mu_2 + \frac{\Delta}{2} \tag{$\mu_1 - \mu_2 = \Delta$} \\
    &> \hat \mu_2(t).
  \end{align*}
  Thus, for all sufficiently large $t$, the greedy algorithm will
  never select arm 2, meaning that $L < \infty$.
  
  Let $A$ be the event that $\{L = \infty\}$ and $B$ be the event that
  \begin{align*}
    \left\{
      \lim_{n \to \infty}\frac{Z_1(n)}{n} = \mu_1
    \right\}
    \cap
    \left\{
      \lim_{n \to \infty}\frac{Z_2(n)}{n} = \mu_2
    \right\}.
  \end{align*}
  We have shown that $A \cap B \subseteq \overline{A}$, meaning $A \cap B =
  \varnothing$.
  Therefore,
  \begin{align*}
    \PX(A)
    &= \PX(A \cap B) + \PX(A \cap \overline{B}) \\
    &= \PX(\varnothing) + \PX(A \cap \overline{B}) \\
    &= \PX(A \cap \overline{B}) \\
    &\le \PX(\overline{B}) \\
    &= 0.
  \end{align*}
\end{proof}

\begin{lemma}
  \label{lem:greedy-min}
  For the greedy algorithm run with initial histories $(S_1(0), S_2(0))$, for
  all $t$,
  \begin{align*}
    \min\left\{\hat \mu_1(t), \hat \mu_2(t)\right\} \le \frac{\min\{S_1(0),
    S_2(0)\}}{N_0}.
  \end{align*}
\end{lemma}
\begin{proof}
  Intuitively, this holds because the greedy algorithm only selects the arm with
  the higher mean, so the lower mean remains unchanged. The min of the two can
  therefore only decrease. We proceed by induction, with the trivial base case
  \begin{align*}
    \min\left\{\hat \mu_1(0), \hat \mu_2(0)\right\}
    &= \min\left\{\frac{S_1(0)}{N_0}, \frac{S_2(0)}{N_0}\right\}.
  \end{align*}
  To show the inductive step, assume without loss of generality that after time
  $t-1$, arm $1$ is the minimizer, i.e., $\hat \mu_1(t-1) \le \hat \mu_2(t-1)$.
  By definition of the greedy algorithm, arm 2 is chosen at time $t$. As a
  result, the mean for arm 1 is unchanged: $\hat \mu_1(t) = \hat \mu_1(t-1)$.
  Therefore,
  \begin{align*}
    \min\left\{\hat \mu_1(t), \hat \mu_2(t)\right\}
    &\le \hat \mu_1(t) \\
    &= \hat \mu_1(t-1) \\
    &= \min\left\{\hat \mu_1(t-1), \hat \mu_2(t-1)\right\}. \tag{By assumption,
    w.l.o.g.}
  \end{align*}
\end{proof}

\begin{lemma}
  \label{lem:lock-in-upper-bound}
For any fixed $\mc H_0$ such that
$\PX(R(\mc H_0,\mc X))>0$,
\begin{align*}
    \lim_{t \to \infty} \PX(\hat \mu_2(t; \mc H_0) < \mu_1 \given
    R) = 1.
  \end{align*}
\end{lemma}
\begin{proof}
  We could prove a slightly weaker version of this statement using the strong
  law of large numbers: $\lim_{t \to \infty} \hat \mu_1(t) = \mu_1$, since arm
  1 will be selected at every $t \ge L$. Thus, if $\hat \mu_2(L) > \mu_1$,
  then eventually the two empirical means will cross at some point,
  contradicting the fact that $L$ was the last switching time.

  This argument would suffice to prove
  \begin{align*}
    \lim_{t \to \infty} \hat \mu_2(t; \mc H_0, \mc X) \given
    R
    &= \hat \mu_2(L; \mc H_0, \mc X) \given
    R \\
    &\le \mu_1
  \end{align*}
  with probability 1.
  We need to strengthen this to a strict inequality. To do so, we rely on the
  law of iterated logarithm \citep[e.g.,][]{kolmogorov-lil}, which states that
  if arm 1 is chosen infinitely often (so $n_1(t) \to \infty$),
  \begin{align*}
    \liminf_{t \to \infty} \frac{\hat \mu_1(t) - \mu_1}{\sqrt{2\mu_1(1-\mu_1)
    \log \log n_1(t) / n_1(t)}} = -1 \quad \text{a.s.}
  \end{align*}
  As a result, $\hat \mu_1(t)$ drops strictly below $\mu_1$ infinitely often
  almost surely. Formally,
  \begin{align*}
    \PX(\exists t > L : \hat \mu_1(t) < \mu_1 \given R) = 1.
  \end{align*}
  Conditioned on $R$, because $\hat \mu_{2}(t) = \hat \mu_{2}(L)$ for
  any $t > L$ (since this is the last time arm $2$ is pulled), with probability
  1, there exists $t$ such that
  \begin{align*}
    \mu_1
    &> \hat \mu_1(t)
    \tag{By the Law of the Iterated Logarithm} \\
    &\ge \hat \mu_{2}(t)
    \tag{Arm 1 is chosen for every $t > L$} \\
    &= \hat \mu_{2}(L) \tag{Arm 2 is never chosen for $t > L$}
  \end{align*}
  Therefore, with probability 1, conditioned on $R$, $\hat \mu_2(L) < \mu_1$.
\end{proof}

\begin{lemma}
  \label{lem:wrong-converge}
  For any $\mc H_0$,
  \begin{align*}
    \lim_{t \to \infty} \PX(W(t))
    &= \PX(W).
  \end{align*}
\end{lemma}
\begin{proof}
  Observe that $W(t) \subseteq W(t+1)$ and
  $\bigcup_{t=1}^{\infty}W(t)=W$. The claim follows from continuity of
  probability from below.
\end{proof}

\begin{lemma}
  \label{lem:blackwell-av-bound}
  Let $\{\hat \theta_t\}_{t=1}^{\infty}$ be the running mean of iid Bernoulli
  trials $\{Y_t\}_{t=1}^\infty$ with true mean $\theta$. Then, for any $t_0$,
  \begin{align*}
    \Pr\b{\exists t > t_0 : \hat \theta_t \le \theta - \sqrt{\frac{2
    \log(1/\delta)}{t_0}}} \le \delta.
  \end{align*}
\end{lemma}
\begin{proof}
  This follows directly from \citet[Thm. 1]{blackwell1997large}, which states
  that for $a,b > 0$,
  \begin{align*}
    \Pr\b{\exists t : \sum_{s=1}^t(\theta-Y_s) \ge a + bt}
    &\le \exp(-2ab) \\
    \Pr\b{\exists t : \theta - \hat \theta_t \ge \frac{a}{t} + b}
    &\le \exp(-2ab) \\
    \Pr\b{\exists t : \hat \theta_t \le \theta - \frac{a}{t} - b}
    &\le \exp(-2ab).
  \end{align*}
  Choose $a = b t_0$ and $b = \sqrt{\log(1/\delta)/(2t_0)}$. This yields
  \begin{align*}
    \Pr\b{\exists t : \hat \theta_t \le \theta -
      t_0 \frac{\sqrt{\log(1/\delta)/(2t_0)}}{t} -
    \sqrt{\frac{\log(1/\delta)}{2 t_0}}}
    &\le \exp\p{- 2t_0 \frac{\log(1/\delta)}{2t_0}} \\
    \Pr\b{\exists t : \hat \theta_t \le \theta -
      t_0 \frac{\sqrt{\log(1/\delta)/(2t_0)}}{t} -
    \sqrt{\frac{\log(1/\delta)}{2 t_0}}}
    &\le \delta
  \end{align*}
  Restricting to $t > t_0$ only decreases the LHS, so
  \begin{align*}
    \Pr\b{\exists t > t_0 : \hat \theta_t \le \theta -
      t_0 \frac{\sqrt{\log(1/\delta)/(2t_0)}}{t} -
    \sqrt{\frac{\log(1/\delta)}{2 t_0}}}
    &\le \delta \\
    \Pr\b{\exists t > t_0 : \hat \theta_t \le \theta -
    \sqrt{\frac{2\log(1/\delta)}{t_0}}}
    &\le \delta
    \tag{$t > t_0$}
  \end{align*}
  as desired.
\end{proof}

\begin{lemma}
  \label{lem:inf-inequality}
  Let $\delta$ be a real-valued random variable such that $\Pr[\delta > 0] = 1$.
  Let $Y_t$ be a sequence of random variables such that $Y_t \asto 0$. Then,
  \begin{align*}
    \lim_{t \to \infty} \Pr[Y_t < \delta] = 1.
  \end{align*}
\end{lemma}
\begin{proof}
  Since $Y_t \to 0$ a.s., $\limsup_{t \to \infty} Y_t = 0$ a.s. Since
  $\delta > 0$ a.s., $\limsup_{t \to \infty} Y_t < \delta$ a.s., and so
  by definition of $\limsup$, $Y_t < \delta$ for all large $t$ a.s.
  Therefore $\mathbf{1}[Y_t < \delta] \xrightarrow[t \uparrow \infty]{} 1$ a.s.,
  and by the Dominated Convergence Theorem,
  \begin{align*}
    \lim_{t \to \infty} \Pr[Y_t < \delta]
    = \lim_{t \to \infty} \mathbb{E}[\mathbf{1}[Y_t < \delta]]
    = 1.
  \end{align*}
\end{proof}

\begin{lemma}
  \label{lem:abcd}
  Let $a,b,c,d$ be nonnegative. Assume
  \begin{align*}
    \frac{a}{b} \ge \frac{c}{d} .
  \end{align*}
  Then, for any $\alpha \in (0, 1)$,
  \begin{align*}
    \frac{\alpha a + c}{\alpha b + d} \le \frac{a + c}{b + d} 
  \end{align*}
\end{lemma}
\begin{proof}
\begin{align*}
    \frac{\alpha a + c}{\alpha b + d}
    &\stackrel{?}{\le} \frac{a + c}{b + d} \\
    (\alpha a + c)(b+d)
    &\stackrel{?}{\le} (\alpha b + d)(a+c) \\
    \alpha a b + \alpha a d + bc + cd
    &\stackrel{?}{\le} \alpha a b + \alpha b c + ad + cd \\
    \alpha a d + bc
    &\stackrel{?}{\le} \alpha b c + ad \\
    bc(1-\alpha) 
    &\stackrel{?}{\le} ad (1-\alpha) \\
    (1-\alpha) \frac{c}{d}
    &\stackrel{\checkmark}{\le} (1-\alpha) \frac{a}{b}.
  \end{align*}
\end{proof}

\end{document}